\documentclass[a4paper]{article}

\usepackage[all]{xy}\usepackage[latin1]{inputenc}        
\usepackage[dvips]{graphics,graphicx}
\usepackage{amsfonts,amssymb,amsmath,color,mathrsfs, amstext,bm}
\usepackage{amsbsy, amsopn, amscd, amsxtra, amsthm,authblk, enumerate}
\usepackage{upref}
\usepackage[colorlinks,
            linkcolor=red,
            anchorcolor=red,
            citecolor=red
            ]{hyperref}

\usepackage{geometry}
\usepackage[displaymath]{lineno}
\usepackage{float}
\usepackage{enumitem}

\usepackage{yhmath}

\usepackage[normalem]{ulem}

\usepackage{graphicx}
\usepackage{subcaption}

\usepackage{algorithm}

\usepackage{algorithmic}

\usepackage{booktabs}

\usepackage{xcolor}

\numberwithin{equation}{section}

\newtheorem*{theorem*}{Theorem}

\newtheorem{lemma}{Lemma}[section]
\newtheorem{assumption}{Assumption}[section]
\newtheorem{proposition}{Proposition}[section]
\newtheorem{remark}{Remark}[section]
\newtheorem{corollary}{Corollary}[section]

\newtheorem{example}{Example}[section]

\usepackage{todonotes}

\DeclareMathOperator{\DKL}{D_{\mathrm{KL}}}
\DeclareMathOperator{\Var}{\mathrm{Var}}

\begin{document}
\title{Guidance for twisted particle filter: a continuous-time perspective}
\author[a]{Jianfeng Lu\thanks{E-mail:jianfeng@math.duke.edu}}
\author[b]{Yuliang Wang\thanks{Email:YuliangWang$\_$math@sjtu.edu.cn}}
\affil[a]{Department of Mathematics, Department of Physics, Department of Chemistry, Duke University, Durham, NC 27708, USA.}
\affil[b]{School of Mathematical Sciences, Institute of Natural Sciences, Shanghai Jiao Tong University, Shanghai, 200240, P.R.China.}

\date{}
\maketitle

\begin{abstract}
The particle filter (PF), also known as sequential Monte Carlo (SMC), approximates high-dimensional probability distributions and their normalizing constants in the discrete-time setting. To reduce the variance of the Monte Carlo approximation, various twisted particle filters (TPFs) have been proposed, in which a twisting function is chosen or learned to modify the Markov transition kernel. Guided by existing control-based importance sampling algorithms in the continuous-time setting, we propose a novel algorithm called the ``Twisted-Path Particle Filter'' (TPPF), in which the twisting function is parameterized by a neural network and trained to minimize a specific KL-divergence between path measures. Numerical experiments illustrate the capability of the proposed algorithm.
\end{abstract}


\begin{keywords}
twisted particle filter, sequential Monte Carlo, stochastic optimal control, importance sampling, neural network, policy gradient
\end{keywords}

\textbf{MSC codes:} 65C05, 93E20, 60G35

\section{Introduction}\label{sec:intro}
The particle filter (PF), or sequential Monte Carlo (SMC), has found wide application across computational statistics and machine learning. It has gained increasing popularity in tasks such as statistical inference for state space models \cite{durbin2012time,sarkka2023bayesian,liu1998sequential,pitt1999filtering} and complex static models \cite{chopin2002sequential,del2006sequential}, and capability and safety pipelines for large language models (LLM) \cite{ziegler2019fine,stiennon2020learning,zou2023universal,zhao2024probabilistic}. Broadly, the particle filter simulates a system of particles over time and is particularly suited for estimating statistical quantities of the following form via Monte Carlo:
\begin{equation}\label{eq:Zintro}
    Z = \mathbb{E}\left[\prod_{k=0}^n g_k(X_k)\right],
\end{equation}
where $g_k$ ($0\leq k\leq n$) are deterministic positive functions and $(X_k)_{k=0}^n$ is a discrete-time Markov chain in $\mathbb{R}^d$ with transition kernel $P(x,dy)$. The simplest and most classical algorithm for approximating \eqref{eq:Zintro} is the bootstrap particle filter (BPF) of Gordon, Salmond, and Smith \cite{gordon1993novel} (see Algorithm \ref{alg:bpf} below). However, to reach a desired precision, the required sample size $N$ can be prohibitive, since the variance of $\prod_{k=0}^n g_k(X_k)$ can be large. To reduce this variance, various ``twisted'' particle filters have been developed \cite{heng2020controlled,guarniero2017iterated,bon2022monte,whiteley2014twisted,branchini2021optimized,zhao2024probabilistic,lawson2018twisted,lawson2022sixo}. The general idea of the twisted particle filter (TPF) is the following: choose a positive twisting function $\varphi(k,x)$ ($1\leq k \leq n$, $x\in \mathbb{R}^d$), run the twisted Markov chain $X^\varphi$ with transition kernel $P^\varphi_k(x,\cdot) \propto \varphi(k,\cdot)P(x,\cdot)$ at step $k$, and modify the function sequence $g_k$ according to $\varphi$ (denoted by $g_k^\varphi$; see \eqref{eq:twistedtransition}--\eqref{eq:ellk}). The twisted model is designed to achieve:
\begin{enumerate}
    \item[(1)] The statistical quantity $Z$ is preserved (i.e. $\mathbb{E}\left[\prod_{k=0}^n g_k(X_k)\right] = \mathbb{E}\left[\prod_{k=0}^n g_k^\varphi(X_k^\varphi)\right]$)
    \item[(2)] The variance of $\prod_{k=0}^n g_k^\varphi(X_k^\varphi)$ is significantly reduced with a suitable choice of $\varphi$.
\end{enumerate}
This is precisely importance sampling \cite{kloek1978bayesian,glynn1989importance,awad2013zero} via a change of measure on discrete path space, and the key step in constructing a twisted model is finding the optimal twisting function that minimizes the variance above. In fact, most existing TPF algorithms adopt a ``look-ahead'' strategy when constructing or learning the twisting function $\varphi$ (see, e.g., \cite{pitt1999filtering,heng2020controlled,guarniero2017iterated}), motivated by the explicit expression of the optimal twisting function (see \eqref{eq:recursive} and \eqref{eq:lookahead} below).

Similar minimization problems are well-studied in the continuous-time setting, where the importance sampling task can be recast as a stochastic optimal control problem. A substantial line of research has developed algorithms for this problem \cite{hartmann2016model,hartmann2017variational,ribera2024improving,richter2021solving}, with applications ranging from molecular dynamics \cite{hartmann2012efficient,ribera2024improving} to mathematical finance \cite{glasserman2004monte,glasserman2005importance}. In this paper, for a continuous-time Markov process $(X_t)_{t\geq 0}$ and deterministic functions $h$, $g$, we are particularly interested in a target quantity of the form
\begin{equation}\label{eq:targetintro2}
    Z = \mathbb{E}\left[\exp\left(\int_0^T h(s,X_s)ds \right) g(X_T)\right],
\end{equation}
This continuous-time framework, and its associated control-based importance sampling algorithms, will serve as guidance for designing novel TPF algorithms in the discrete-time setting.

Although the discrete-time particle filter and continuous-time control-based importance sampling are typically studied in separate communities, they share a common mathematical structure. In this paper, we build a bridge between the two by setting up both models side by side in Section \ref{sec:ISsetting}. Appendix \ref{sec:convergence} provides a rigorous convergence proof from the discrete-time model to the continuous-time model under suitable assumptions. Based on this connection, we propose a novel algorithm---the ``Twisted-Path Particle Filter'' (TPPF)---directly guided by a family of control-based importance sampling algorithms and their theoretical foundations, including the classical Donsker-Varadhan variational principle \cite{boue1998variational,deuschel2001large,hartmann2017variational}. We expect TPPF to be more versatile, less problem-dependent, and better-behaved in high-dimensional state spaces. Section \ref{sec:num} validates TPPF numerically and compares it with existing approaches.

\subsection{Related works}\label{sec:relatework}

\paragraph{Twisted particle filters (TPF)} The idea of twisting the model in the discrete-time setting is well established, and various twisted particle filter algorithms have been proposed across computational statistics and machine learning. In \cite{pitt1999filtering}, the authors proposed the so-called ``fully-adapted auxiliary particle filter'' (FA-APF), applying the ``look-ahead'' strategy to twist the model with $g_k(\cdot)$ as the twisting function at each discrete time $k$. In \cite{guarniero2017iterated}, the authors proposed the ``iterated auxiliary particle filter'' (iAPF) algorithm, which uses the recursive relation of the optimal twisting function (see \eqref{eq:recursive} below) and learns the twisting function iteratively via Galerkin approximation with Gaussian basis functions. Similar approaches were adopted in \cite{heng2020controlled}, where the authors studied the problem from the viewpoint of discrete-time optimal control and improved the learning structure for better numerical stability. This fixed-point iteration method for approximating the optimal twisting function \cite{guarniero2017iterated,heng2020controlled} is similar in spirit to the temporal difference (TD) method \cite{sutton1988learning,doya1995temporal,mou2024bellman} widely used in the reinforcement learning (RL) community. The same iteration structure has been applied and improved in other works such as \cite{bon2022monte, lawson2018twisted}: \cite{bon2022monte} uses rejection sampling to generate samples from the twisted Markov chain, while \cite{lawson2018twisted} uses neural networks to parameterize the twisting function in place of Galerkin approximation. Beyond TD-type methods, the authors of \cite{branchini2021optimized} proposed an ``optimized auxiliary particle filter'' (OAPF) algorithm, which solves a convex optimization problem at each discrete time $k$ to update the weights and positions of particles while simulating the Markov chain.

\paragraph{Learning-based TPF algorithms} Most TPF algorithms above are based on optimization over a Galerkin function space. More recently, motivated by applications to large language models (LLM) in the machine learning community, several deep-learning-based structures for approximating the optimal twisting function have been proposed. These learning-based algorithms include the Contrastive Twist Learning (CTL) in \cite{zhao2024probabilistic}, the Future Discriminators for Generation (FUDGE) in \cite{yang2021fudge}, and the SI$\mathcal{X}$O method in \cite{lawson2022sixo}. In these frameworks, the authors parameterize the twisting function by a neural network, choose the sum of some specially designed quantities along the time axis as the loss function (for instance, in \cite{zhao2024probabilistic}, the KL-divergence at each time $k$ determined by the current twisting function), and train the network to approximate the optimal twisting function. We discuss the choice of loss function further in Section \ref{sec:TPPF}. Compared with Galerkin-based methods, such neural network approximations tend to scale better to high dimensions.

\paragraph{Continuous-time importance sampling (IS) via stochastic optimal control} As mentioned above, this paper studies the (discrete-time) particle filters based on some well-known control-based importance sampling algorithms. From the theoretical perspective, control-based importance sampling algorithms are well-studied in the importance sampling community. Various (continuous-time) importance sampling algorithms have been proposed, such as \cite{kloek1978bayesian,glynn1989importance,awad2013zero}. We are particularly interested in importance sampling algorithms for the path-dependent target \eqref{eq:targetintro2}, and focus on a sequence of control-based ones due to the zero-variance property related to the solution of the stochastic optimal control problem. A large number of control-based importance sampling algorithms admit the structure of the continuous-time policy gradient (PG) method \cite{williams1992simple} (also called the iterative diffusion optimization (IDO) in some other literature \cite{richter2021solving}), and the corresponding details will be discussed later in Section \ref{sec:motivation}. There are various classical choices for the loss function in the PG iteration in the existing literature. In the ``cross-entropy'' algorithm \cite{zhang2014applications,hartmann2016model}, the loss is chosen as the KL-divergence between path measures $\DKL(P^{u^{*}} \| P^{u})$, where $P^u$ is the path measure induced by the controlled SDE adding a control $u$ to the drift, and $P^{u^{*}}$ is the path measure corresponding to the optimal control $u^{*}$. Note that different from some other statistical inference problems, under the current setting, $\DKL(P^{u^{*}} \| P^{u})$ has an explicit expression and is even quadratic after suitable parameterization of the control function $u$. Other choices of the loss function include the relative entropy loss $\DKL(P^{u} \| P^{u^{*}})$ \cite{hartmann2017variational,lie2016strongly,ribera2024improving}, the variance loss (or log-variance loss) $\Var_{P^v}(\tfrac{dP^{u^{*}}}{dP^u})$ (or $\Var_{P^v}(\log\tfrac{dP^{u^{*}}}{dP^u})$) for some suitable basis path measure $P^v$ \cite{richter2021solving}, etc. Some theoretical bounds for the KL-type losses above were established in \cite{hartmann2024nonasymptotic}. Besides the PG-based algorithms, other related importance sampling methods include the well-known forward-backward stochastic differential equation (FBSDE) approaches \cite{kebiri2019adaptive,fuhrman2010stochastic,wang2018introduction}, where one approximates the target value $Z$ via the solution of some SDE with given terminal-time state and a forward filtration.

\paragraph{Learning-based importance sampling (IS) algorithms} Many of the conventional methods above may suffer from the curse of dimensionality. In order to solve the related stochastic optimal control problem discussed above, a wide range of learning-based algorithms are available in reinforcement learning for the continuous-time setting. In these frameworks, the learning objective is usually parameterized by neural networks, and over decades a variety of these algorithms have been proposed and studied, such as the classical (soft) policy gradient \cite{sutton1999policy,williams1992simple}, actor-critic with temporal difference learning \cite{zhou2024solving,mou2024bellman} or with deep backward SDE \cite{han2018solving}, soft Q-learning \cite{haarnoja2017reinforcement,schulman2017equivalence}, etc.
Other improvements to these learning-based frameworks include adding entropy regularization to the optimal control problem \cite{wang2020reinforcement}, combining the metadynamics algorithms when $X_s$ in \eqref{eq:targetintro2} is trapped in some metastable region \cite{ribera2024improving}, and applying them to discrete-observed continuous-time filtering problems \cite{chopin2023computational}.

\subsection{Main contributions}\label{sec:contribution}

In this work, we provide novel insights into the study of particle filters. We begin by building a continuous-time model, which can be viewed as a continuous limit of the discrete-time model of interest. Guided by classical control-based importance sampling algorithms for the continuous-time model, we design new TPF algorithms in the discrete-time setting, including our TPPF algorithm. This continuous-time perspective also opens promising future directions for the design of TPF algorithms.

Notably, TPPF takes the KL divergence between path measures as its loss function and learns the optimal twisting function via neural networks. From a reinforcement learning viewpoint, our approach is a policy-gradient method, whereas many existing approaches follow the temporal-difference paradigm. Compared with Galerkin-based approximations of the optimal twisting function, TPPF overcomes the curse of dimensionality in many models. Moreover, because its learning procedure is problem-independent---unlike popular Galerkin-based methods such as \cite{heng2020controlled,guarniero2017iterated,bon2022monte}---TPPF is more robust and applies to a wider range of models.



The rest of the paper is organized as follows. Section \ref{sec:ISsetting} introduces the basic settings of both the discrete- and continuous-time cases, discusses importance sampling for each model via the twisting function or control variate respectively, and presents the existence and zero-variance property of the optimal twisting function and optimal control. Motivated by existing algorithms for the continuous-time model, Section \ref{sec:TPPF} proposes the novel particle filter algorithm---the ``Twisted-Path Particle Filter'' (TPPF). Section \ref{sec:num} reports several numerical examples comparing TPPF with existing approaches. Section \ref{sec:conclusion} concludes and outlines possible future work. The Appendix collects proofs of technical lemmas and propositions, along with a rigorous convergence analysis from the discrete- to the continuous-time model in Appendix \ref{sec:convergence}.

\section{Importance sampling for discrete-time and continuous-time models}\label{sec:ISsetting}
In this section, we introduce both discrete-time and continuous-time models, as well as importance sampling methods based on change of measure for both models. 

\subsection{Notation}
We first introduce the notation used in the rest of the article. Given integers $n_1\leq n_2$ and a sequence $(x_k)_{k \in \mathbb{N}}$, we write the subsequence $x_{n_1:n_2} := (x_{n_1},\dots,x_{n_2})$. For the discrete-time model, to be introduced in Section \ref{sec:dissetting} below, we denote the target quantity, the discrete Markov chain, the transition kernel, and the associated deterministic function sequence by $Z_{dis}$, $\hat{X}$, $\hat{P}$, and $\hat{g}$. We also denote the twisting function by $\varphi(k,x)$ and the optimal twisting function by $\varphi^{*}(k,x)$ ($k \in \mathbb{N}$, $1 \leq k \leq n$, $x\in \mathbb{R}^d$). For twisting function $\varphi$, we denote by $\hat{X}^\varphi$ the twisted Markov chain, by $\hat{P}^\varphi$ the twisted transition kernel, and by $\hat{P}[\varphi](k,x) := \int \varphi(k,y')\hat{P}(x,dy')$ the normalizing constant associated with $\hat{P}^\varphi_k(x,dy) \propto_y \varphi(k,y) \hat{P}(x,dy)$.
For the continuous-time model, to be introduced in Section \ref{sec:consetting} below, we denote the target quantity, the continuous Markov process, the transition kernel, and the associated deterministic functions by $Z_{con}$, $X$, $P$ and ($h$, $g$). For a control variate $u(t,x)$, we denote by $X^u$ the solution to the controlled SDE.
In the proposed TPPF algorithm, we denote by $P^{\varphi}$ the (discrete) path measure induced by the twisting function $\varphi$, and by $r(P^\varphi)$ the relative variance associated with $P^\varphi$.

\subsection{Discrete-time model and optimal twisting}\label{sec:dissetting}
Let us begin with the discrete-time model. 
Fix a positive integer $n$.  Given a (time-homogeneous) Markov transition kernel $\hat{P}(x,dy)$ in $\mathbb{R}^d$ (for simplicity we assume the corresponding transition density exists and still denote it by $\hat{P}(\cdot,\cdot)$), a sequence of bounded, continuous, nonnegative functions $\hat{g}_k(\cdot)$ ($0 \leq k \leq n$), and a (deterministic) initial point $x \in \mathbb{R}^d$, the general discrete-time model is determined by the triple $(\hat{P}(\cdot,\cdot), (\hat{g}_k(\cdot))_k, x)$: For the discrete Markov chain $\hat{X}_{0:n}$ with transition kernel $\hat{P}(\cdot,\cdot)$ and initial state $\hat{X}_0 = x$, the statistical quantity of interest is given by
\begin{equation}\label{eq:Zdisdef}
    Z_{dis}(x) := \mathbb{E}_x\left[\prod_{k=0}^n \hat{g}_k(\hat{X}_k)\right],
\end{equation}
where $\mathbb{E}_x\left[\cdot\right] := \mathbb{E}\left[ \cdot \mid \hat{X}_0 = x\right]$.

Such a general model is often called the Feynman-Kac model \cite{del2004feynman,bon2022monte}, and the quantity $Z_{dis}(x)$ corresponds to the terminal marginal measure in it. Note that throughout this paper, we consider the Markov chain with a homogeneous transition kernel and a fixed deterministic initial state. As is common in the literature, such settings are only designed to simplify the arguments and notations, and all results in this paper can be easily extended to the Markov chains with inhomogeneous transition kernels and with random initial conditions.

\begin{remark}[state space model]\label{rmk:statespacedis}
A special case of the general model above is the Hidden Markov Model (a state space model with discrete observations), where the functions $\hat{g}_{0:n}(\cdot)$ are determined by the random observation values $y_{1:n}$ with an observation conditional density function $\hat{g}_{y|x}^k(\cdot\mid \hat{X}_k)$ at each time $k$:
\begin{equation}
    \hat{X}_i\mid \hat{X}_{i-1} \sim \hat{P}(\hat{X}_{i-1}, \cdot), \quad \hat{Y}_i\mid \hat{X}_{i} \sim \hat{g}_{y|x}^i(\cdot\mid \hat{X}_i),\quad 0\leq i \leq n.
\end{equation}
Then, given the observations $y_{0:n}$, the functions $\hat{g}_{0:n}(\cdot)$ are given by
\begin{equation}
    \hat{g}_k(x) := \hat{g}^k_{y|x}(y_k \mid x).
\end{equation}
Consequently, conditioning on the observations $y_{0:n}$, the distribution of $\hat{X}_{1:n}$ is proportional to
\begin{equation*}
    \hat{g}_0(x_0)\prod_{k=1}^n \hat{g}_k(x_k) \hat{P}(x_{k-1},x_k) dx_{1:n},
\end{equation*}
and its normalizing constant is $Z_{dis}(x)$ defined in \eqref{eq:Zdisdef}. 
\end{remark}

Now, with the triple $(\hat{P}(\cdot,\cdot), (\hat{g}_k(\cdot))_k, x)$ of the discrete model, we aim to calculate $Z_{dis}(x)$. A standard approach is the particle filter method (also known as sequential Monte Carlo) \cite{doucet2009tutorial,djuric2003particle,doucet2000sequential,doucet2001introduction,liu1998sequential}, where the expectation in $Z_{dis}(x)$ is simulated via Monte Carlo. A classical particle filter algorithm is the following bootstrap particle filter (BPF) \cite{gordon1993novel}:
\begin{algorithm}[H]
	\renewcommand{\algorithmicrequire}{\textbf{Input:}}
	\renewcommand{\algorithmicensure}{\textbf{Output:}}
	\caption{Bootstrap particle filter (BPF)}
	\label{alg:bpf}
	\begin{algorithmic}
            \STATE Given $(\hat{P}(\cdot,\cdot), (\hat{g}_k(\cdot))_k, x)$ and particle number $N$. Set $\zeta_0^i = x$, $1\leq i \leq N$.
            \STATE For $k = 1,\dots,n$, 
            \STATE 1. Calculate the particle weights
            \begin{equation}\label{eq:bpfweight}
                W_{k-1}^j =  \dfrac{\hat{g}_{k-1}(\zeta_{k-1}^j)}{\sum_{\ell=1}^N \hat{g}_{k-1}(\zeta_{k-1}^\ell)},\quad 1 \leq j \leq N.
            \end{equation}
            \STATE 2. Sample independently
            \begin{equation}\label{eq:bpfresample}
                \zeta_k^i \sim \sum_{j=1}^N W_{k-1}^j\hat{P}(\zeta^j_{k-1},\cdot) ,\quad 1\leq i \leq N.
            \end{equation}
		\ENSURE  $\zeta_k^i$ ($k=0,\dots,n$, $i = 1,\dots, N$), $Z^N_{dis}(x) (:= \prod_{k=0}^n \frac{1}{N}\sum_{i=1}^N \hat{g}_k(\zeta_k^i))$.
	\end{algorithmic}  
\end{algorithm}

Note that for each $k$, in \eqref{eq:bpfresample}, particles are resampled according to the weights $W_{k-1}^{1:N}$ defined in \eqref{eq:bpfweight}.
An effective way to implement \eqref{eq:bpfresample} is the ancestor-prediction method (see for instance \cite{guarniero2017iterated, bon2022monte}), namely, at $k$-th step,
\begin{enumerate}
    \item Sample ancestors $A_{k-1}^i \sim \mathcal{C}(W_{k-1}^1, \dots, W_{k-1}^N)$, $1\leq i \leq N$.
    \item Sample predictions $\zeta_k^i|\zeta^{A_{k-1}^i}_{k-1} \sim \hat{P}(\zeta^{A_{k-1}^i}_{k-1},\cdot)$, $1 \leq i \leq N$.
\end{enumerate}
Above, $\mathcal{C}(\cdot,\dots,\cdot)$ denotes the categorical distribution to sample the indexes. Moreover, the resampling step does not need to occur at each step, and one improvement is the $\kappa$-adapted resampling \cite{kong1994sequential,liu1995blind,del2012adaptive}, where the resampling step only occurs when the effective sample size is smaller than $\kappa N$ for some fixed $\kappa \in (0,1)$. 

For any $w^{1:N}$ with $\sum_{i=1}^N w^i = 1$, the effective sample size $ESS(w^{1:N})$ is defined by
\begin{equation}
    ESS(w^{1:N}) := 1 / \sum_{i=1}^N (w^i)^2,
\end{equation}
and it is a good criterion for evaluating a particle filter algorithm---higher ESS usually means better performance.

\medskip 

A shortcoming of BPF is the relatively high variance, especially when the dimension $d$ or the length of the Markov chain $n$ is large.
The twisted particle filter was introduced \cite{heng2020controlled,guarniero2017iterated,bon2022monte,whiteley2014twisted,branchini2021optimized,zhao2024probabilistic,lawson2018twisted,lawson2022sixo} to reduce the variance of the computed $Z^N_{dis}(x)$ based on importance sampling.  The general idea is to use the (discrete-time) change of measure, or equivalently, a sequence of twisting functions $\varphi(k,\cdot)$ ($1\leq k \leq n$). Then, we consider the twisted model with the triple $(\hat{P}^\varphi(\cdot,\cdot), (\hat{g}_k^\varphi(\cdot))_k, x)$ defined by
\begin{equation}\label{eq:twistedtransition}
    \hat{P}^\varphi_k(x,dy) := \frac{\varphi(k,y) }{\hat{P}[\varphi](k,x)}\hat{P}(x,dy),\quad\hat{P}[\varphi](k,x) := \int \varphi(k,y') \hat{P}(x,dy')\quad 1\leq k\leq n,
\end{equation}
\begin{equation}\label{eq:twistedgk}
    \hat{g}^{\varphi}_k(x) := \hat{g}_k(x) \ell^{\varphi}_k(x),\quad 0\leq k\leq n,
\end{equation}
with
\begin{equation}\label{eq:ellk}
    \ell_0^{\varphi}(x) := \hat{P}[\varphi](1,x),\quad \ell_n^\varphi(x) := \frac{1}{\varphi(n,x)},\quad \ell^{\varphi}_k(x) := \frac{\hat{P}[\varphi](k+1,x)}{\varphi(k,x)},\quad 1\leq k\leq n-1.
\end{equation}
Then it is not difficult to check that the twisted model preserves the quantity \eqref{eq:Zdisdef}. We conclude this property in the following proposition (see Appendix \ref{app:proofsec2} for a detailed proof via a straightforward calculation):
\begin{lemma}\label{eq:discretegirsanov}
Consider the twisted model defined in \eqref{eq:twistedtransition}--\eqref{eq:ellk}. Recall the quantity $Z_{dis}(x)$ defined in \eqref{eq:Zdisdef}. Then
\begin{equation}
    Z_{dis}(x) = \mathbb{E}_x\left[\prod_{k=0}^n\hat{g}_k^\varphi(\hat{X}^\varphi_k)\right],
\end{equation}
where $\hat{X}^\varphi$ is the twisted Markov chain with the (time-inhomogeneous) transition kernel $\hat{P}_k^{\varphi}(\cdot,\cdot)$, namely, $\hat{X}^\varphi_{k} \sim \hat{P}^\varphi_k(\hat{X}^\varphi_{k-1}, \cdot)$ ($1\leq k \leq n$). 
\end{lemma}

\begin{remark}[discrete-time change of measure]\label{rmk:discretegirsanov}
Lemma \ref{eq:discretegirsanov} is in fact the Girsanov's transform in the discrete-time setting. Denote $\hat{\mathcal{P}}$ the law of the untwisted Markov chain $\hat{X}_{0:n}$ with the transition kernel $\hat{P}(\cdot,\cdot)$, and denote $\hat{\mathcal{P}}^\varphi$ the law of the twisted Markov chain $\hat{X}^\varphi_{0:n}$ with the transition kernel $\hat{P}^\varphi(\cdot,\cdot)$. Then, the Girsanov transform gives
\begin{equation}
    Z_{dis}(x) = \mathbb{E}_x^{X_{0:n}\sim \hat{\mathcal{P}}}\left[\prod_{k=0}^n \hat{g}_k(X_k)\right] = \mathbb{E}_x^{X_{0:n}\sim \hat{\mathcal{P}}^\varphi}\left[\prod_{k=0}^n \hat{g}_k(X_k) \frac{d\hat{\mathcal{P}}}{d\hat{\mathcal{P}}^\varphi}(X)\right],
\end{equation}
where the Radon-Nikodym derivative is given by
\begin{equation*}
    \frac{d\hat{\mathcal{P}}}{d\hat{\mathcal{P}}^\varphi}(X) = \prod_{k=0}^n \ell_k^\varphi(X_k) = \prod_{k=1}^n \frac{\hat{P}[\varphi](k,X_{k-1})}{\varphi(k,X_k)}.
\end{equation*}
    
\end{remark}

With the twisted model \eqref{eq:twistedtransition} - \eqref{eq:ellk} and Lemma \ref{eq:discretegirsanov}, it is natural to consider the following twisted particle filter (TPF) method \cite{heng2020controlled,guarniero2017iterated,bon2022monte,whiteley2014twisted,branchini2021optimized,zhao2024probabilistic,lawson2018twisted,lawson2022sixo}, whose output $Z^{N,\varphi}_{dis}(x)$ defined below approximates our target quantity $Z_{dis}(x)$ due to the law of large numbers (see for instance, Proposition 3 in \cite{guarniero2017iterated}, or Section 3.6 in \cite{doucet2009tutorial}):
\begin{algorithm}[H]
	\renewcommand{\algorithmicrequire}{\textbf{Input:}}
	\renewcommand{\algorithmicensure}{\textbf{Output:}}
	\caption{Twisted particle filter (TPF)}
	\label{alg:tpf}
	\begin{algorithmic}
            \STATE Given $\varphi(\cdot,\cdot)$, $(\hat{P}^{\varphi}(\cdot,\cdot), (\hat{g}_k^{\varphi}(\cdot))_k, x)$ and particle number $N$. Set $\zeta_0^i = x$, $1\leq i \leq N$.
            \STATE For $k = 1,\dots,n$,
            \STATE 1. Calculate the particle weights
            \begin{equation}\label{eq:tpfweight}
                W_{k-1}^j =  \dfrac{\hat{g}^\varphi_{k-1}(\zeta_{k-1}^j)}{\sum_{\ell=1}^N \hat{g}^\varphi_{k-1}(\zeta_{k-1}^\ell)},\quad 1\leq j \leq N.
            \end{equation}
            \STATE 2. Sample independently
            \begin{equation}\label{eq:tpfresample}
                \zeta_k^i \sim \sum_{j=1}^N W_{k-1}^j\hat{P}^\varphi(\zeta^j_{k-1},\cdot) ,\quad 1\leq i \leq N.
            \end{equation}
		\ENSURE  $\zeta_k^i$ ($k=0,\dots,n$, $i = 1,\dots, N$), $Z^{N,\varphi}_{dis}(x) (:= \prod_{k=0}^n \frac{1}{N}\sum_{i=1}^N \hat{g}^{\varphi}_k(\zeta_k^i))$.
	\end{algorithmic}  
\end{algorithm}

Now since the $\varphi$-twisted model (recall its definition \eqref{eq:twistedtransition}--\eqref{eq:ellk}) obtains an unbiased estimation for the quantity $Z_{dis}(x)$ by Lemma \ref{eq:discretegirsanov}, it is natural to ask the following question: How do we choose $\varphi$ to minimize the variance?
In fact, it is possible to choose an optimal twisting function sequence $\varphi^{*}(k,\cdot)$ ($1\leq k \leq n$), under which the random variable $\prod_{k=0}^n\hat{g}_k^{\varphi^{*}}(\hat{X}^{\varphi^{*}}_k)$ is an unbiased estimate of $Z_{dis}(x)$ with zero variance. Optimal twisting functions have been widely studied in the literature \cite{bon2022monte,heng2020controlled,guarniero2017iterated,branchini2021optimized}. In particular, the optimal twisting function sequence $\varphi^{*}(k,\cdot)$ ($1\leq k \leq n$) is defined recursively as follows:
\begin{equation}\label{eq:recursive}
\begin{aligned}
    &\varphi^{*}(k,x) = \hat{g}_k(x) \int \varphi^{*} (k+1,y) \hat{P}(x,dy),\quad 0 \leq k \leq n-1,\\
    &\varphi^{*}(n,x) = \hat{g}_n(x).
\end{aligned}
\end{equation}
Moreover, we can prove the Feynman-Kac-like formula for $\varphi^{*}$:
\begin{equation}\label{eq:lookahead}
    \varphi^{*}(k,x) = \mathbb{E}\left[\prod_{i=k}^n \hat{g}_{i}(\hat{X}_{i}) \mid \hat{X}_{k} = x\right], \quad 0 \leq k \leq n, \quad x \in \mathbb{R}^d,
\end{equation}
and under the optimal twisting function $\varphi^{*}$, the random variable $\prod_{k=0}^n\hat{g}_k^\varphi(\hat{X}^\varphi_k)$ has zero variance. Consequently, the output of the TPF algorithm under $\varphi^{*}$ has zero variance, namely, $Z^{N,\varphi^{*}}_{dis}(x) = Z_{dis}(x)$. We provide more details and the rigorous proof in the Appendix.

Note that although the $\varphi^{*}$-twisted particle filter provides a perfect approximation for our target $Z_{dis}(x)$, it is impossible to implement the twisted particle filter (Algorithm \ref{alg:tpf}) associated with $\varphi^{*}$ in practice. The reason is that we need the pointwise value of the twisting function so that $\hat{X}^{\varphi^{*}}_{k}$ can be sampled from the distribution proportional to $\varphi^{*}(k,\cdot)\hat{P}(\hat{X}_{k-1}^{\varphi^{*}},\cdot)$. Therefore, it is crucial to find suitable approaches to approximate the optimal twisting functions $\varphi^{*}(k,\cdot)$ ($1\leq k \leq n$) to lower the variance of the particle filter algorithm. Moreover, due to the law of large numbers \cite{doucet2009tutorial,guarniero2017iterated}, a larger sample size $N$ and a better approximation of $\varphi^{*}$ give a better approximation for $Z_{dis}(x)$ (i.e. smaller variance of the numerical output $Z_{dis}^{N,\varphi}(x)$).

In Section \ref{sec:TPPF}, we will propose a new method to approximate $\varphi^{*}$, guided by the similar method for the continuous-time model described in the following subsection.

\subsection{Continuous-time model and optimal control}\label{sec:consetting}

Now let us introduce the continuous-time model. As mentioned in the introduction, the purpose of setting up this continuous-time model is:  It can be viewed as a continuous limit of the discrete-time model in Section \ref{sec:dissetting} we aim to study. Consequently, some more well-studied algorithms for the continuous-time model can inspire novel (discrete-time) TPF algorithms. 

For a fixed function $b: \mathbb{R}^d \rightarrow \mathbb{R}^d$, we consider the following SDE
\begin{equation}\label{eq:SDE}
    X_t = X_0 + \int_0^t b(X_s) ds + \sqrt{2}B_t,\quad X_0 = x\in\mathbb{R}^d.
\end{equation}
where $(B_t)_{t\geq 0}$ is the Brownian motion in $\mathbb{R}^d$ under the probability measure $P$.
The dynamics is thus a time-homogeneous Markov process with transition density (or Green's function) $P_t(y|x)$ satisfying a Fokker-Planck equation
\begin{equation}\label{eq:FPcon}
    \partial_t P_t(y|x) = -\nabla_y \cdot (b(y) \, P_t(y|x)) + \Delta_y P_t(y|x),\quad P_0(y|x) = \delta_x(y),
\end{equation}
where $\delta_x(\cdot)$ is the Dirac delta at $x$. Moreover, given functions $h: \mathbb{R}_{+} \times \mathbb{R}^d \rightarrow \mathbb{R}$, $g: \mathbb{R}^d \rightarrow \mathbb{R}$ ($g$ is nonnegative) and $T>0$, the general continuous model is determined by $(P;h,g;x)$, and the corresponding statistical quantity of interest is
\begin{equation}\label{eq:zcon}
    Z_{con}(x) = \mathbb{E}_x\left[e^{\int_0^T h(s,X_s) ds}g(X_T)\right],
\end{equation}

Now consider the general continuous model determined by $(P;h,g;x)$. Similarly to the discrete-time model, we focus on variance reduction via control-based importance sampling, which is based on a change of measure, as discussed in the discrete case. In detail, the change of measure is constructed by adding a control variate $u$ to the drift and considering the following controlled SDE:
\begin{equation}\label{eq:controlledSDE}
    X^u_t = X_0^u + \int_0^t \left(b(X_s^u) + \sqrt{2}u(s,X_s^u) \right) ds + \sqrt{2}B_t, \quad X_0^u = x \in\mathbb{R}^d.
\end{equation}
Then, denoting the path measures $\mathcal{P} := \text{Law}(X_{[0,T]})$ ($X$ solving \eqref{eq:SDE}) and $\mathcal{P}^u := \text{Law}(X^u_{[0,T]})$ ($X^u$ solving \eqref{eq:controlledSDE}), the change of measure gives
\begin{equation}
\begin{aligned}
    Z_{con}(x) &= \mathbb{E}_x^{X_{[0,T]}\sim \mathcal{P}}\left[e^{\int_0^T  h(s,X_s) ds} g(X_T)\right]\\
    &= \mathbb{E}_x^{X_{[0,T]}\sim \mathcal{P}^u}\left[e^{\int_0^T  h(s,X_s) ds} g(X_T) \frac{d\mathcal{P}}{d\mathcal{P}^u}(X)\right] = \mathbb{E}_x\left[e^{\int_0^T  h(s,X^u_s) ds} g(X^u_T)\frac{d\mathcal{P}}{d\mathcal{P}^u}(X^u)\right],
\end{aligned}
\end{equation}
where the Radon-Nikodym derivative has the following expression due to Girsanov's theorem:
\begin{equation}
    \frac{d\mathcal{P}}{d\mathcal{P}^u}(X^u) = \exp \left(-\int_0^T u(s,X^u_s) \cdot dB_s - \frac{1}{2}\int_0^T |u(s,X^u_s)|^2 ds \right).
\end{equation}
Our goal is to minimize the variance of the random variable $e^{\int_0^T  h(s,X^u_s) ds} g(X^u_T)\frac{d\mathcal{P}}{d\mathcal{P}^u}(X^u)$ for $u$ belonging to some suitable admissible set. Luckily, we can find an optimal control (similar to the discrete case) satisfying a zero-variance property, and then we only need to find suitable approaches to approximate the optimal control. In fact, the optimal control $u^{*}$ is given by
\begin{equation}
    u^{*}(t,x) = \sqrt{2}\partial_x \log v^{*}(t,x),
\end{equation}
where $v^{*}(t,x)$ satisfies the following backward Kolmogorov (parabolic) equation:
\begin{equation}\label{eq:optimalu}
    \begin{aligned}
        &-\partial_t v^{*}(t,x) = b(x) \cdot \nabla_x v^{*}(t,x) + \Delta_x v^{*}(t,x) + h(t,x) v^{*}(t,x), \quad 0 \leq t \leq T,\\
        &v^{*}(T,x) = g(x).
    \end{aligned}
\end{equation}
Moreover, we can show that under the control $u^{*}$, the random variable $$e^{\int_0^T  h(s,X^{u^{*}}_s) ds} g(X^{u^{*}}_T)\frac{d\mathcal{P}}{d\mathcal{P}^{u^{*}}}(X^{u^{*}})$$ 
has zero variance. We provide more details and a rigorous proof in the Appendix.
In Section \ref{sec:TPPF}, we will discuss detailed methods to find this $u^{*}$, mainly based on the classical Donsker-Varadhan variational principle.

\begin{remark}[Connection between the discrete-time and continuous-time models]
As a final remark, under suitable settings, the continuous-time model in Section \ref{sec:consetting} is a continuous limit of the discrete-time model in Section \ref{sec:dissetting}. To see this, we rewrite $Z_{dis}(x)$ by
\begin{equation}
    Z_{dis}(x) = \mathbb{E}_x \left[\exp\left(\sum_{k=0}^{n-1} \int_{k\eta}^{(k+1)\eta} \eta^{-1} \log \hat{g}_k(\hat{X}_k) ds\right) \hat{g}_n(\hat{X}_n)\right],\quad \forall \eta > 0,
\end{equation}
which shares a similar structure with \eqref{eq:zcon}. A rigorous proof is given in Appendix \ref{sec:convergence}, mainly based on the Feynman-Kac representation. Therefore, we can design novel particle filters for the discrete-time model guided by some more well-studied algorithms for the continuous-time model. We give more details in the next section.

\end{remark}

\section{Twisted-Path Particle Filter (TPPF)}\label{sec:TPPF}

In this section, motivated by mature algorithms from the continuous-time community, we propose a Twisted-Path Particle Filter (TPPF). In this new algorithm for the discrete-time model, we treat the KL-divergence between the twisted path measure and the optimal path measure as the loss function, parameterize the twisting function via neural networks, and approximate the optimal twisting function via suitable optimization methods such as stochastic gradient descent. Compared with other existing TPF algorithms, most of which are Galerkin-based, we expect this proposed algorithm to (1) perform better in high dimensions; (2) be more robust and have wider applications (in particular, be applicable to many nonlinear, non-Gaussian models); (3) have a relatively stronger theoretical foundation. We also remark that we only focus on the dynamical setting for the particle filter in our numerical examples, rather than static cases, such as the approach of annealed importance sampling (see for instance, Section 2.4 of \cite{heng2020controlled}).

\subsection{Motivation: importance sampling via variational characterization in continuous-time setting}\label{sec:motivation}
In the continuous-time model, it has been relatively well-studied to find the optimal control variate from a variational perspective. In order to seek the optimal control variate with the zero-variance property, instead of solving the PDE we derived in \eqref{eq:FPcon}, one can convert this to an optimal control problem from a variational perspective \cite{hartmann2017variational,zhang2014applications,ribera2024improving,hartmann2016model}. We first give some background on the so-called Donsker-Varadhan variational principle \cite{boue1998variational,deuschel2001large,hartmann2017variational}, which is independent of whether the model is time-continuous or time-discrete. For the reader's convenience, we provide a proof of Lemma \ref{lmm:DV} in the Appendix.

\begin{lemma}[Donsker-Varadhan variational principle]\label{lmm:DV}
Given some function $W(\cdot): \Omega \rightarrow \mathbb{R}$ and probability measure $P$ on $\Omega$, the following relation holds:
\begin{equation}\label{eq:DVdual}
    -\log\mathbb{E}^{X \sim P}\left[\exp\left(-W(X) \right)\right] = \inf_{Q \in \mathcal{P}(\Omega)}\bigl\{ \mathbb{E}^{X\sim Q}\left[W(X)\right] + \DKL(Q \| P)\bigr\},
\end{equation}
where the minimum is over all probability measures on $\Omega$.
\end{lemma}

Take $\Omega$ to be the path space corresponding to the continuous-time model and choose $W(\cdot)$ to be the path integral for some stochastic process $(X_s)_{0\leq s\leq T}$, as in \cite{hartmann2017variational,zhang2014applications,ribera2024improving,hartmann2016model}
\begin{equation}
    W(X) := -\int_0^T h(s,X_s) ds - \log g(X_T)
\end{equation}
for some $h(\cdot,\cdot)$, $g(\cdot)$ defined in Section \ref{sec:consetting}, the variational relation \eqref{eq:DVdual} becomes
\begin{multline}
    -\log\mathbb{E}^{X_{[0,T]}\sim P}\left[\exp \left(\int_0^T h(s,X_s) ds + \log g(X_T) \right)\right]\\
    =\inf_{Q\in \mathcal{P}(\Omega),Q\ll P}\mathbb{E}^{X_{[0,T]} \sim Q}\left[\int_0^T -h(s,X_s) ds - \log g(X_T) \right] + \DKL(Q\|P).
\end{multline}
In particular, if the path measure $P$ above is the law of some $X_{[0,T]}$ satisfying an SDE as introduced in Section \ref{sec:consetting}
\begin{equation}\label{eq:uncontrolled}
    dX_t = b(X_t)dt + \sqrt{2}dB, \quad 0 \leq t \leq T, \quad X_0 = x,
\end{equation}
then the probability $Q$ above is characterized by the path measure of the controlled SDE after adding a control variate $u$ to the drift. In detail, consider the controlled SDE
\begin{equation}\label{eq:controlled}
    dX^u_t = b(X^u_t)dt + \sqrt{2}u(t,X_t^u)dt + \sqrt{2}dB, \quad 0 \leq t \leq T, \quad X^u_0 = x.
\end{equation}
By Girsanov's theorem, the dual relation \eqref{eq:DVdual} becomes
\begin{multline}
    -\log \mathbb{E}_x\left[\exp \left(\int_0^T h(s,X_s) ds + \log g(X_T) \right)\right]\\
    = \inf_{u\in \mathcal{U}}\mathbb{E}_x\left[-\int_0^T h(s,X^u_s) ds + \frac{1}{2}\int_0^T |u(s,X_s^u)|^2 ds - \log g(X_T^u)\right],
\end{multline}
and under the setting of Section \ref{sec:consetting}, the set of admissible controls $\mathcal{U}$ is usually chosen to be as follows (see for instance Section 1 in \cite{richter2021solving})
\begin{equation}
    \mathcal{U}=\left\{u \in C^1\left(\mathbb{R}^d \times[0, T], \mathbb{R}^d\right):  u \text { grows at most linearly in } x \right\}.
\end{equation}
Denote
\begin{equation}
    J(u) := \mathbb{E}_x\left[-\int_0^T h(s,X^u_s) ds + \frac{1}{2}\int_0^T |u(s,X_s^u)|^2 ds - \log g(X_T^u)\right].
\end{equation}
It can be verified that the minimizer $u^{*}$ of the functional $J(u)$ yields a zero-variance importance sampler for the continuous setting, namely,
\begin{equation}
    \exp \left(\int_0^T h(s,X^{u^{*}}_s) ds + \log g(X^{u^{*}}_T) \right) \frac{dP}{dP^{u^{*}}}(X^{u^{*}}) = Z_{con}(x),\quad P^{u^{*}} - a.s.,
\end{equation}
where $P^{u^{*}}$ is the path measure induced by the controlled process $X^{u^{*}}$ in \eqref{eq:controlled} associated with the optimal control $u^{*}$, and $P$ is the path measure induced by the original process $X$ in \eqref{eq:uncontrolled}. Therefore, it is reasonable to find the optimal control variate by solving the following stochastic optimal control problem:
\begin{equation}\label{eq:ju}
    \min_{u \in \mathcal{U}} J(u) = \mathbb{E}_x\left[-\int_0^T h(s,X^u_s) ds + \frac{1}{2}\int_0^T |u(s,X_s^u)|^2 ds - \log g(X_T^u)\right]. 
\end{equation}
In practice, one may parameterize $u(t,x)$ by $u = u(\theta;t,x)$ and iteratively solve the optimization problem. In detail, given a suitable loss function $L(u)$ that depends on the control $u$ and the trajectory $X^u_{[0,T]}$ (e.g., $L = J$), at each iteration one implements the following:
\begin{enumerate}
    \item With the current control $u(t,x) = u(\theta;t,x)$ simulate $N$ realizations of the controlled SDE $X^u_{[0,T]}$, and calculate the loss $L$ and its derivative $\nabla_\theta L$ using the $N$ realizations of $X^u_{[0,T]}$.
    \item Update $\theta$ using $\nabla_\theta L$ via some suitable method, for instance, the stochastic gradient descent.
\end{enumerate}
As mentioned in Section \ref{sec:intro}, such an iterative framework is exactly the continuous-time policy gradient (PG) method \cite{williams1992simple}, and it is also called the iterative diffusion optimization (IDO) in some other literature \cite{richter2021solving}. As a remark, $J(u)$ of the form \eqref{eq:ju} offers a good choice of the loss $L$, which equals $\DKL(P^u\|P^{u^{*}})$ up to a constant. Moreover, others have also studied other forms of loss to approximate the optimal control $u^{*}$; for instance, in the so-called cross-entropy method \cite{zhang2014applications,hartmann2016model}, they use the loss $\DKL(P^{u^{*}}\|P^u)$ instead. More discussion on the choice of loss function and detailed derivations would be given in the next subsection for the TPPF algorithm, where we also treat the KL-divergence between path measures as the loss function.

\subsection{Approximating the optimal twisting function in discrete-time setting}
Now we propose a novel framework to approximate the optimal twisting function $\varphi^{*}$ defined in \eqref{eq:recursive}, guided by the existing method for training $u^{*}$ in the continuous-time setting. Different from the iteration method proposed in \cite{heng2020controlled,guarniero2017iterated,lawson2018twisted} (which is similar to the temporal difference (TD) learning in the reinforcement learning community \cite{sutton1988learning,doya1995temporal,mou2024bellman,zhou2024solving}) based on the recursive formula in \eqref{eq:recursive}, our method learns the optimal twisting function along the whole path via a neural network. Hence, we name the proposed algorithm the ``Twisted-Path Particle Filter'' (TPPF).

In what follows, we first derive a similar variational principle and some relations between different choices of loss induced by the twisting function in the discrete-time setting. After that, we will give our detailed algorithms corresponding to possible choices of loss function. Some detailed formulas and implementation details will also be discussed at the end of this section.

First, note that the Donsker-Varadhan variational principle can also be applied to the discrete-time model in Section \ref{sec:dissetting}. In fact, after a change of measure, we have
\begin{equation}
    -\log Z_{dis}(x) = \inf_{\varphi} J(\varphi),
\end{equation}
where $\varphi = \left(\varphi_1,\dots,\varphi_n \right)$, $\varphi_i(\cdot)$ ($1 \leq i \leq n$) are positive, continuous functions in $\mathbb{R}^d$, and as a direct result of the dual relation, by choosing $W(\hat{X}) = -\sum_{k=0}^n\log \hat{g}_k(\hat{X}_k)$ for some discrete Markov chain $(\hat{X}_k)_{k=0}^n$
\begin{equation}\label{eq:Jvarphi}
    J(\varphi) := \mathbb{E}_x\left[-\sum_{k=0}^n\log \hat{g}_k(\hat{X}_k^\varphi)\right] + \DKL(P^{\varphi}\|P^1),
\end{equation}
where $P^{\varphi}$ is the path measure induced by the twisting function $\varphi$ as in \eqref{eq:recursive}, $P^1$ is the original path measure (or equivalently, with twisting function being $1$), and $\hat{X}^\varphi$ is the twisted Markov chain associated with $P^\varphi$. More precisely, under $P^{\varphi}$, the Markov chain $X^{\varphi}$ evolves under the twisted transition density
\begin{equation}
    \hat{P}^{\varphi}_k(x_{k-1},x_k) = \frac{\varphi(k,x_k)\hat{P}(x_{k-1},x_k)}{\hat{P}[\varphi](k,x_{k-1})},\quad 1 \leq k \leq n,
\end{equation}
with the normalizing constant $\hat{P}[\varphi](k,x_{k-1})$ defined by
\begin{equation*}
    \hat{P}[\varphi](k,x_{k-1}) = \int \varphi(k,y)\hat{P}(x_{k-1},y)dy,
\end{equation*}
and under $P$, the Markov chain $X$ evolves under the untwisted transition density $\hat{P}(x_{k-1},x_k)$. Consequently, the twisted path measure associated with the twisting function $\varphi$ has the following explicit expression:
\begin{equation}\label{eq:Pvarphi}
    P^\varphi(dx_{1:n}) = \prod_{k=1}^n \frac{\varphi(k,x_k)\hat{P}(x_{k-1},x_k)}{\hat{P}[\varphi](k,x_{k-1})} dx_{1:n}.
\end{equation}

Similarly to the continuous-time setting, it can be verified that the optimal twisting function $\varphi^{*}$ with the zero-variance property is the minimizer of the functional $J(\varphi)$ above, and minimizing the functional $J(\varphi)$ is equivalent to minimizing the KL-divergence $\DKL(P^{\varphi}\|P^{\varphi^{*}})$. We summarize this result in the proposition below:
\begin{proposition}\label{eq:jvarphiexpression}
Consider the functional $J(\varphi)$ defined in \eqref{eq:Jvarphi} and the optimal twisting function $\varphi^{*}$ defined in \eqref{eq:recursive}. Then it holds
\begin{equation}
    J(\varphi) = J(\varphi^{*}) + \DKL(P^{\varphi}\|P^{\varphi^{*}}).
\end{equation}
Moreover, one can derive the following explicit formula for $J(\varphi)$ and $\DKL(P^{\varphi}\|P^{\varphi^{*}})$:
\begin{equation}
    J(\varphi) = \mathbb{E}_x\left[-\sum_{k=0}^n \log \hat{g}_k(\hat{X}_k^{\varphi}) + \sum_{k=1}^n \log \frac{\varphi(k,\hat{X}_k^{\varphi})}{ \hat{P}[\varphi](k,\hat{X}^{\varphi}_{k-1})}\right]
\end{equation}
and
\begin{equation}
    \DKL(P^{\varphi}\|P^{\varphi^{*}}) = \mathbb{E}_x\left[-\sum_{k=0}^n \log \hat{g}_k(\hat{X}_k^{\varphi}) + \sum_{k=1}^n \log \frac{\varphi(k,\hat{X}_k^{\varphi})}{ \hat{P}[\varphi](k,\hat{X}^{\varphi}_{k-1})}\right] + \log \varphi^{*}(0,x)
\end{equation}
\end{proposition}

\begin{remark}
After a change of measure, it is easy to see that $J(\varphi)$ or $\DKL(P^{\varphi}\|P^{\varphi^{*}})$ has an alternative expression:
\begin{multline}\label{eq:untwisted}
    \DKL(P^{\varphi}\|P^{\varphi^{*}}) = \mathbb{E}_x\Big[\exp\Big(\sum_{k=1}^n\log \frac{\varphi(k,\hat{X}_k)}{\hat{P}[\varphi](k,\hat{X}_{k-1})} \Big)\\
    \Big(-\sum_{k=0}^n \log \hat{g}_k(\hat{X}_k) + \sum_{k=1}^n \log \frac{\varphi(k,\hat{X}_k)}{\hat{P}[\varphi](k,\hat{X}_{k-1})}\Big)\Big] + \log \varphi^{*}(0,x).
\end{multline}
This expression is particularly useful in cases where sampling from the original transition $\hat{P}(x_{k-1},x_k)$ is much easier than sampling from the twisted transition $\hat{P}_k^{\varphi}(x_{k-1},x_k)$. For instance, if $\varphi$ is parameterized by some neural network, basic sampling methods like rejection sampling (proposed in \cite{bon2022monte} for the twisted particle filter) would not be so efficient, suffering from a low, unstable acceptance rate. Moreover, it is natural to doubt whether the Monte Carlo approximation for the loss (and its gradient) would deteriorate after a change of measure. So far we have no theoretical guarantee for this, but empirically, the experiments show that \eqref{eq:untwisted} can approximate the loss well. 
\end{remark}

Guided by various variational-based methods in continuous-time setting (discussed in Section \ref{sec:motivation}), we can also consider different loss functions other than $J(\varphi)$ or $\DKL(P^{\varphi}\|P^{\varphi^{*}})$. Before seeking blindly for other possible choices, let us first study the relationship between the relative variance and KL-divergence, since our final goal of approximating the optimal twisting function is just to reduce the variance. In fact, using a generalized Jensen's inequality, we are able to derive the following:
\begin{proposition}\label{prop:KLcontrolinequality}
Given some function $W(\cdot): \Omega \rightarrow \mathbb{R}$ and probability measures $P$, $Q$ on $\Omega$ which are absolutely continuous with each other, define
\begin{equation*}
    Z := \mathbb{E}^{X \sim P}\left[e^{-W(X)}\right] = \mathbb{E}^{\tilde{X}\sim Q}\left[e^{-W(\tilde{X})}\frac{dP}{dQ}(\tilde{X})\right]
\end{equation*}
and the relative variance with respect to $Q$ is defined by
\begin{equation}\label{eq:Relvar}
    r(Q) :=\frac{\sqrt{Var_{Q}\left(e^{-W}\frac{dP}{dQ} \right)}}{Z}
\end{equation}
Suppose there is an optimal probability measure $Q^{*}$ with the zero-variance property
\begin{equation}
    \frac{dQ^{*}}{dP} = \frac{e^{-W}}{Z} \quad P-a.s.,
\end{equation}
Then the following estimates hold:
\begin{enumerate}
    \item $r^2(Q) \geq e^{\DKL(Q^{*}\|Q)}-1$
    \item If the constants $m:=\inf_E \frac{Q^{*}(E)}{Q(E)}$, $M = \sup_E \frac{Q^{*}(E)}{Q(E)}$ exist, then
    \begin{equation}
        e^{m\DKL(Q \| Q^{*}) + \DKL(Q^{*} \| Q)}-1 \leq r^2(Q) \leq e^{M\DKL(Q \| Q^{*}) + \DKL(Q^{*} \| Q)}-1
    \end{equation}
\end{enumerate}
\end{proposition}

Consider the discrete-time model with $(\hat{P}(\cdot,\cdot); \hat{g}_{0:n}(\cdot);x)$, we have the following direct corollary of Proposition \ref{prop:KLcontrolinequality}.
\begin{corollary}\label{coro:coro}
Take $\Omega = \mathbb{R}^n$, $W(\hat{X}) = -\sum_{k=0}^n\log \hat{g}_k(\hat{X}_k)$, $P = P^1$, $Q = P^\varphi$, and $Q^{*} = P^{\varphi^{*}}$ (recall the definition of $P^\varphi$ in \eqref{eq:Pvarphi}) in Proposition \ref{prop:KLcontrolinequality}. Then there exists $0< m \leq M$ depending on $\varphi$, $\varphi^{*}$ such that
\begin{equation}
        e^{m\DKL(P^\varphi \| P^{\varphi^{*}}) + \DKL(P^{\varphi^{*}} \| P^\varphi)}-1 \leq r^2(P^\varphi) \leq e^{M\DKL(P^\varphi \| P^{\varphi^{*}}) + \DKL(P^{\varphi^{*}} \| P^\varphi)}-1.
    \end{equation}
\end{corollary}

Motivated by Proposition \ref{prop:KLcontrolinequality} and Corollary \ref{coro:coro} above, it is then reasonable to consider the loss function of the form 
\begin{equation*}
    a\DKL(P^\varphi \| P^{\varphi^{*}}) + \DKL( P^{\varphi^{*}} \| P^\varphi)
\end{equation*}
for some positive $a$. Moreover, when $P^\varphi$ approximates the target $P^{\varphi^{*}}$ well, $m$, $M$ in Proposition \ref{prop:KLcontrolinequality} are approximately $1$, so it is reasonable to choose $a=1$, namely, treat 
\begin{equation*}
    \DKL(P^\varphi\| P^{\varphi^{*}}) + \DKL(P^{\varphi^{*}} \| P^\varphi)
\end{equation*}
as a loss function. Of course with these bounds, it is also reasonable to consider the loss
\begin{equation*}
    \DKL(P^{\varphi^{*}} \| P^\varphi),
\end{equation*}
which corresponds to the loss in the well-studied cross-entropy method in the continuous-time setting \cite{zhang2014applications} (and this is the reason why we denote it by $L_{CE}$). When implementing the TPPF algorithm, we need explicit expressions for the loss functions chosen above, so that we can calculate them via Monte Carlo approximation. In fact, denoting
\begin{equation}\label{eq:denoteL}
L_{RE}:=\DKL(P^\varphi \| P^{\varphi^{*}}),\quad L_{CE} := \DKL(P^{\varphi^{*}} \| P^\varphi),\quad L_{RECE} = L_{RE}+L_{CE}.
\end{equation}
Simple calculations yield
\begin{equation}\label{eq:LRE}
\begin{aligned}
    L_{RE} 
    & =\mathbb{E}_x\left[-\sum_{k=0}^n \log \hat{g}_k(\hat{X}_k^{\varphi}) + \sum_{k=1}^n \log \frac{\varphi(k,\hat{X}_k^{\varphi})}{\hat{P}[\varphi](k,\hat{X}_{k-1}^{\varphi})}\right] + \log \varphi^{*}(0,x)\\
    &= \mathbb{E}_x\left[\exp\left(\sum_{k=1}^n\log \frac{\varphi(k,\hat{X}_k)}{\hat{P}[\varphi](k,\hat{X}_{k-1})} \right)\left(-\sum_{k=0}^n \log \hat{g}_k(\hat{X}_k) + \sum_{k=1}^n \log \frac{\varphi(k,\hat{X}_k)}{\hat{P}[\varphi](k,\hat{X}_{k-1})}\right)\right] \\
    & \qquad + \log \varphi^{*}(0,x) ,  
\end{aligned}
\end{equation}
where we can use either the first line (via the twisted Markov chain $\hat{X}^\varphi$) or the second line (via the untwisted Markov chain $\hat{X}$) in the training. See Section \ref{sec:num} for more details. Also, for $L_{CE}$, we have
\begin{equation}\label{eq:LCE}
\begin{aligned}
    L_{CE}  &= \frac{1}{\varphi^{*}(0,x)}\mathbb{E}_x\left[\exp\left(\sum_{k=0}^n \log \hat{g}_k(\hat{X}_k) \right) \left(\sum_{k=1}^n \log \hat{g}_k(\hat{X}_k) - \log \varphi^{*}(0,x) \right)\right]  (=\text{constant})\\
    &\quad - \frac{1}{\varphi^{*}(0,x)}\mathbb{E}_x\left[\exp\left(\sum_{k=0}^n \log \hat{g}_k(\hat{X}_k) \right) \left(\sum_{k=1}^n \log \frac{\varphi(k,\hat{X}_k)} {\hat{P}[\varphi](k,\hat{X}_{k-1})} \right)\right].
\end{aligned}
\end{equation}
While it seems that the objective function contains $\varphi^{*}(0,x)$ (besides the constant term, which does not impact the optimization), this factor is of course not known in practice, as $\varphi^{*}(0,x)$ is in fact the target $Z_{dis}(x)$. However, since $\varphi^{*}(0,x)$ is positive by definition, one can instead optimize the factor following $1 / \varphi^*(0, x)$.

\begin{remark}[possible numerical instability of $L_{CE}$]\label{rmk:instability}
Note that the coefficient 
$$\exp\left(\sum_{k=0}^n \log \hat{g}_k(\hat{X}_k) \right) / \varphi^{*}(0,x)$$
in the second line of \eqref{eq:LCE} may bring numerical instability to $L_{CE}$-training.
This coefficient above is determined by $\hat{g}_k$ and the (random) position of $\hat{X}_k$, so it is numerically unstable especially when the length of the Markov chain $n$ is large. This seems to explain some of the instability behavior we observe empirically in numerical examples, see details in Section \ref{sec:num}. 
\end{remark}



Given the loss functions above, to approximate the optimal twisting function $\varphi^{*}$ in practice, it remains to parameterize $\varphi$ and learn the parameters according to the loss functions. In our algorithm, we parameterize the twisting function by a neural network $\varphi(k,x) = \varphi(\theta;k,x)$, where $\theta$ denotes all the parameters of the network, and $k$, $x$ are the inputs. For numerical stability, it is better to treat $\log \varphi(\theta;k,x)=NN(\theta;k,x)$ as the output so that $\varphi(k,x) = \exp(NN(\theta;k,x))$. We refer to Section \ref{sec:num} for more details. We summarize the proposed TPPF in Algorithm \ref{alg:tppf} below:
\begin{algorithm}
	\renewcommand{\algorithmicrequire}{\textbf{Input:}}
	\renewcommand{\algorithmicensure}{\textbf{Output:}}
	\caption{Twisted-Path Particle Filter (TPPF)}
	\label{alg:tppf}
	\begin{algorithmic}
            \STATE Given discrete-time model with $(\hat{P}(\cdot,\cdot), (\hat{g}_k(\cdot))_k, x)$. Parameterize the twisting function via neural network: $\log \varphi(\theta;k,x)=NN(\theta;k,x)$. Choose the loss function $L = L_{RE}$ or $L_{CE}$ or $L_{RECE}$.
            \STATE 1. \textbf{Learn the twisting function.} At each iteration for updating $\theta$,
            \STATE \quad (1) Simulate $N$ independent realizations of the twisted Markov chain $\hat{X}_{0:n}^\varphi$ with transition kernel $\hat{P}^\varphi$ (or alternatively, simulate the untwisted Markov chain $\hat{X}_{0:n}$ with transition kernel $\hat{P}$).
            \STATE \quad (2) Calculate the loss $L$ and its gradient $\nabla_{\theta}L$ according to \eqref{eq:denoteL} - \eqref{eq:gradient} via Monte Carlo approximation using the $N$ samples obtained in (1).
            \STATE \quad (3) Update the parameters $\theta$ using suitable optimization methods (for instance, $\theta \leftarrow \theta - \eta \nabla_{\theta} L$).
            \STATE 2. \textbf{Run the twisted particle filter} (Algorithm \ref{alg:tpf}) using the learned twisting function $\varphi(\theta;k,x)$
	\end{algorithmic}  
\end{algorithm}

\begin{remark}[computing the gradient]
As a final remark, recall that when implementing Algorithm \ref{alg:tppf}, we need to calculate the gradient $\nabla_\theta L$ via Monte Carlo approximation. Note that when the loss function $L$ is approximated via the untwisted Markov Chain $\hat{X}$, the gradient can be calculated via auto-differentiation. However, if the loss $L$ is approximated via the twisted Markov Chain $\hat{X}^\varphi$ (see the first line in \eqref{eq:LRE}), the computation of $\nabla_\theta L$ is not that direct, since $\hat{X}^\varphi$ depends on the parameter $\theta$. Here, we provide the explicit formula for $\partial_{\theta_i} L_{RE}$, where $L_{RE}$ is given in the first line of \eqref{eq:LRE}. In fact, we first write $L_{RE}$ into the form of the second line in \eqref{eq:LRE}, where $\hat{X}$ in it is independent of $\theta$. Then, after a change of measure, we can write the gradient into the following:
\begin{equation}\label{eq:gradient}
\partial_{\theta_i}L_{RE} 
    =\mathbb{E}_x\left[\left(1-\sum_{k=1}^n \log \hat{g}_k(\hat{X}_k^{\varphi}) + \sum_{k=1}^n \log \frac{\varphi(k,\hat{X}_k^{\varphi})}{\hat{P}[\varphi](k,\hat{X}_{k-1}^{\varphi})}\right)\left(\sum_{k=1}^n \partial_{\theta_i}\log \frac{\varphi(k,\hat{X}_k^{\varphi})}{\hat{P}[\varphi](k,\hat{X}_{k-1}^{\varphi})} \right)\right]   .
\end{equation}
For a further implementation detail, to obtain $\partial_{\theta_i}L_{RE}$ in \eqref{eq:gradient}, there is no need to calculate $\partial_{\theta_i}\varphi$. Instead, we can first calculate the value of
\begin{equation*}
    \mathbb{E}_x\left[\left(1-\sum_{k=1}^n \log \hat{g}_k(\hat{X}_k^{\varphi}) + \sum_{k=1}^n \log \frac{\varphi(k,\hat{X}_k^{\varphi})}{\hat{P}[\varphi](k,\hat{X}_{k-1}^{\varphi})}\right)\left(\sum_{k=1}^n \log \frac{\varphi(k,\hat{X}_k^{\varphi})}{\hat{P}[\varphi](k,\hat{X}_{k-1}^{\varphi})} \right)\right]   
\end{equation*}
using Monte Carlo approximation, and then ``detach'' the first term above so that it would contain no gradient information. Consequently, auto-differentiation via (one-time) back-propagation gives the desired value of the gradient in \eqref{eq:gradient}. 
\end{remark}

\section{Numerical examples}\label{sec:num}
In this section, we test the proposed TPPF algorithm with loss functions $L_{RE}$, $L_{CE}$, or $L_{RECE}$ on different models including the linear Gaussian model, the NGM-78 model, and the Lorenz-96 model. We compare our algorithm with well-known competitors including the bootstrap particle filter (BPF) \cite{gordon1993novel}, the iterated auxiliary particle filter (iAPF) \cite{guarniero2017iterated}, and the fully-adapted auxiliary particle filter (FA-APF) \cite{pitt1999filtering}. 
We observe that TPPF is more robust and can beat its competitors in most examples, although it does not perform as well as iAPF in some linear models, mainly due to the linear structure of iAPF.
We provide some implementation details in the Appendix, including details for parameterization and calculation of the normalizing constant $\hat{P}[\varphi](k,x)$, as well as some discussion on the fairness of the comparison.\footnote{The source code in this paper can be found at \url{https://github.com/superwyl666/TPPF}}

\subsection{Linear Gaussian model}\label{sec:LGM}
As a first example, we consider the linear Gaussian model, which is also the discretization of an OU process $X_t$ with linear Gaussian observations $Y_k \sim N(\cdot;B_k X_k,\Sigma_{OB})$
\begin{equation*}
    X_{k+1} = X_k + \Delta t A X_k  + \sqrt{\Delta t} N(0,\Sigma). 
\end{equation*}
We choose $\Delta t = 0.01$, $T = 0.5$ (so the length of the discrete Markov chain is $n = T/\Delta t = 50$), $B_k = I_d$, $A = -I_d$, $\Sigma = I_d$, $\Sigma_{OB} = I_d$.

We consider $d \in \{2,5,15,20 \}$. The mean $\mu_k$ and the variance $\sigma^2_k$ are parameterized by two independent 2-layer DenseNets, each with width 10. We optimize the two networks together using ADAM with learning rate 0.001. The particle number is set to be 200. The boxplots in Figure \ref{fig:boxplotlg} compare the TPPF (with $L_{RE}$, $L_{CE}$, or $L_{RECE}$) with competitors including BPF, iAPF and FA-APF using 1000 replicates. The red cross represents the mean and the red dashed line represents the median. We also report the empirical standard deviations in Table \ref{table:gaussSD}, and in Table \ref{table:combined} (a) we report the average relative effective sample size (ESS-r) defined by $ESS(W_{1:N}) = 1 / (N\sum_{i}^N W_i^2)$ with $\sum_{i=1}^N W_i = 1$. Furthermore, recall the definition of the relative variance (which is the quantity we aim to reduce)
\begin{equation*}
    r(Q) =\frac{\sqrt{Var_{Q}\left(e^{-W}\frac{dP}{dQ} \right)}}{Z} = \sqrt{\left(\frac{\mathbb{E}_Q|e^{-W} \frac{dP}{dQ}|^2}{|\mathbb{E}_Q[e^{-W} \frac{dP}{dQ}]|^2} \right) - 1}.
\end{equation*}
We compute the relative variance of each $Q$ using Monte Carlo approximation with a relatively large sample size, and the result is reported in Table \ref{table:combined} (b). We remark here that as we observe in our experiments, the relative variance is a criterion more sensitive to the nonlinearity of the model. Also, the experiments for the relative variance adopt a smaller value of $T$ to avoid blow-up of data, because compared with the quantity $Z$ ($=\mathbb{E}_P [e^{-W}]$), the computation of $r(Q)$ defined above requires its second moment, where numerical instability is more likely to occur. In particular, we choose $T = 0.1$ in the linear Gaussian model when computing the relative variance.

Clearly, TPPF with $L_{RE}$ and $L_{RECE}$ can defeat BPF and FA-APF, especially when the dimension is high, and TPPF with $L_{CE}$ suffers from the curse of dimensionality due to the numerical instability discussed in Remark~\ref{rmk:instability}. Moreover, the iAPF performs the best in this linear Gaussian model, because iAPF is learning the twisting function in a Gaussian function class, each time solving a standard restricted least square minimization problem. Therefore, our training-based method cannot perform as well as iAPF in this experiment. However, as we will see in the other experiments when the optimal twisting function is not in the Gaussian family, our algorithm performs better than iAPF.
\begin{figure}[htbp]
\centering
\includegraphics[width=6cm]{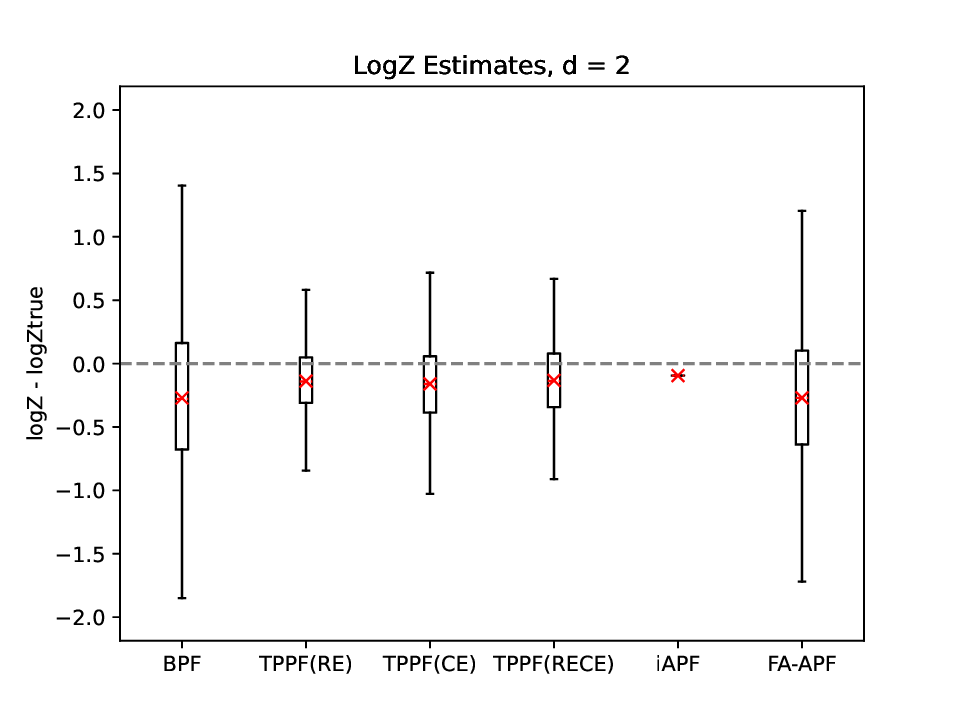}
\quad
\includegraphics[width=6cm]{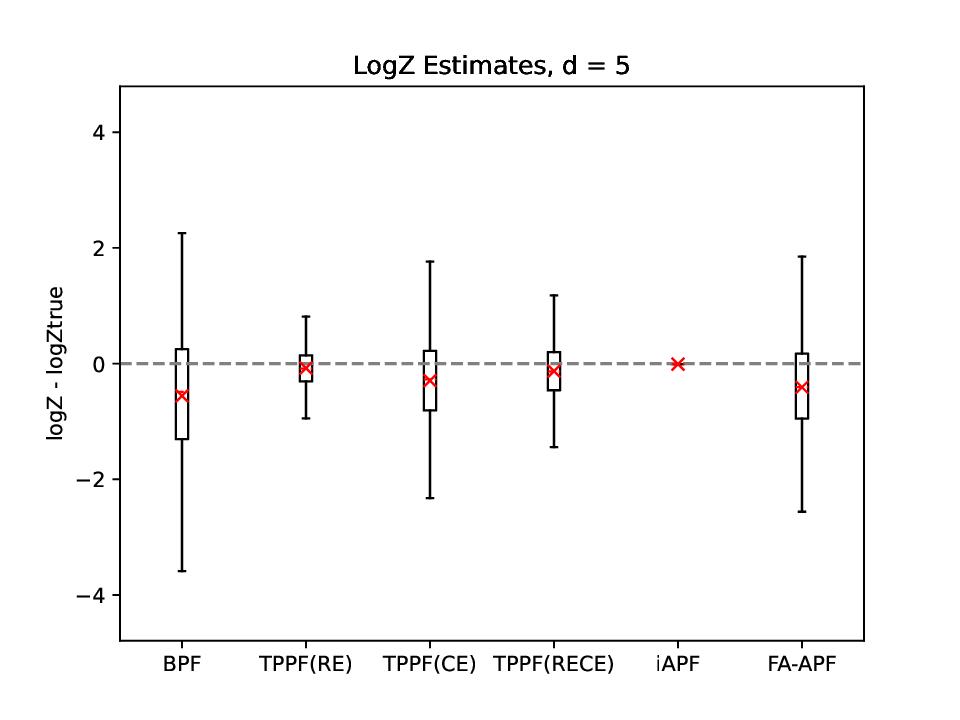}
\quad
\includegraphics[width=6cm]{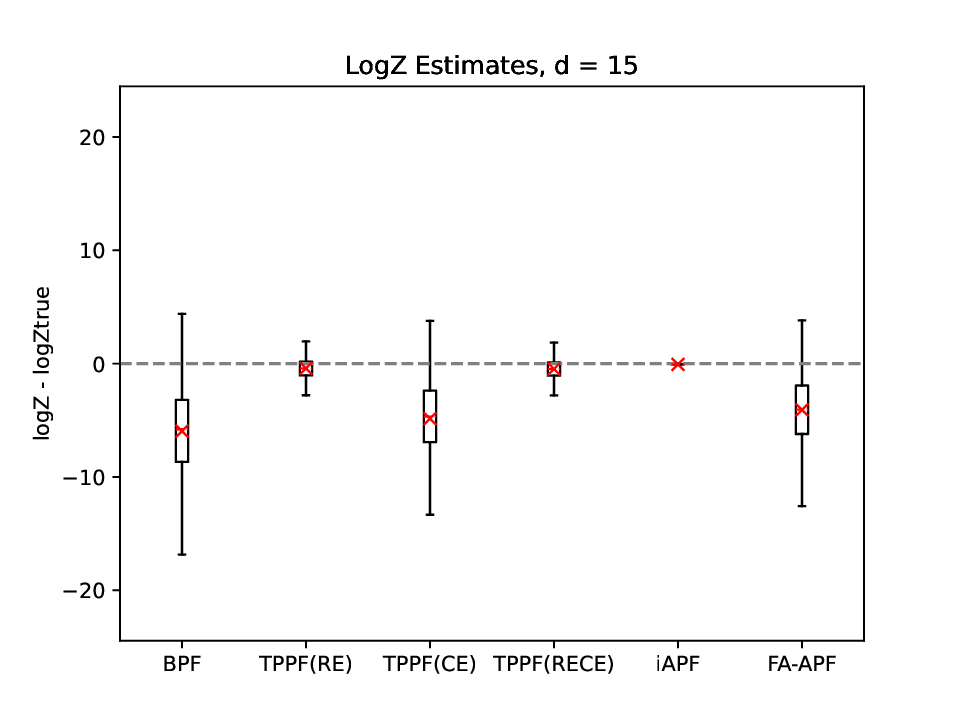}
\quad
\includegraphics[width=6cm]{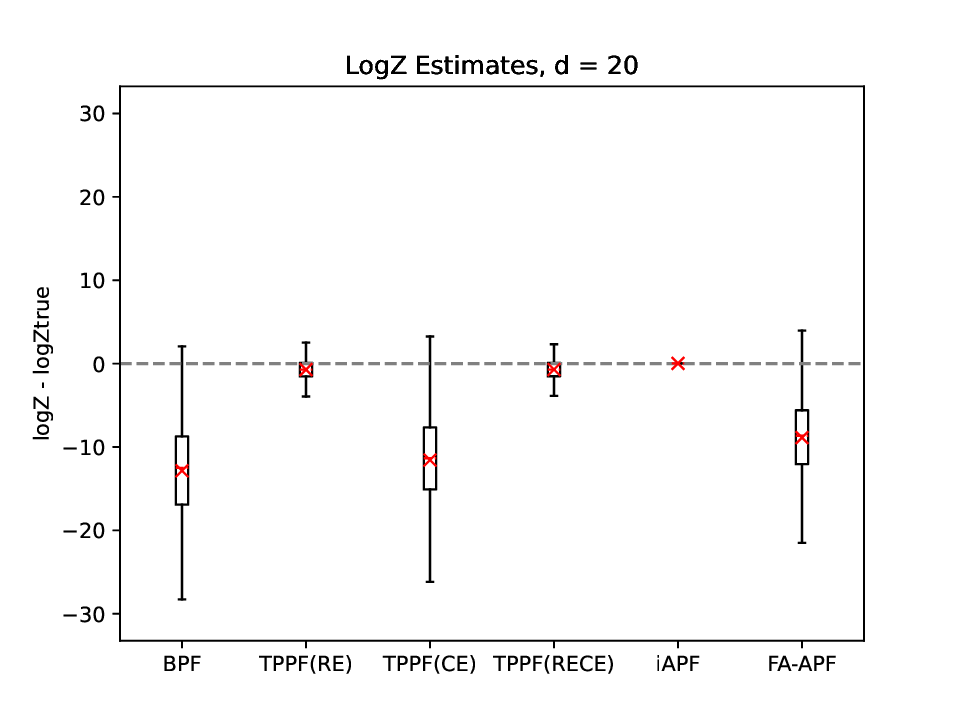}
\caption{\textbf{Linear Gaussian model:} compare TPPF (trained with $L_{RE}$, $L_{CE}$, or $L_{RECE}$) and its competitors (BPF, iAPF and FA-APF). Boxplot for $\log Z$ using 1000 replicates, with configurations $d \in \{2,5,15,20 \}$. The red cross represents the mean and the red dashed line represents the median. }
\label{fig:boxplotlg}
\end{figure}


\begin{table}[H]
    \centering
\begin{tabular}{ccccc}
\toprule  
 &d=2 & d=5  &  d=15  & d=20\\ 
\midrule 
BPF &0.60 & 1.14 & 4.05 & 5.95\\
TPPF(RE) &0.27 & 0.35 & 0.90 & 1.23\\
TPPF(CE) &0.34 & 0.77 & 3.51 &5.70\\
TPPF(RECE) &0.31 & 0.50 & 0.86 &1.21\\
FA-APF &0.61 & 0.87 & 3.13 &4.85\\
iAPF &6.11e-13 & 2.54e-14 & 4.47e-14 &9.43e-14\\
\bottomrule     
\end{tabular}
    \caption{\textbf{Linear Gaussian model:} compare TPPF (trained with $L_{RE}$, $L_{CE}$, or $L_{RECE}$) and its competitors (BPF, iAPF and FA-APF). Empirical standard deviation of $\log Z$ with 1000 replicates for $d \in \{2,5,15,20 \}$.}
    \label{table:gaussSD}
\end{table}


\begin{table}[H]
\centering
\begin{subtable}{0.60\linewidth}
\centering
\begin{tabular}{ccccc}
\hline
 & d=2 & d=5  & d=15 & d=20\\
\hline
BPF & 88.93\% & 73.93\% & 48.18\% & 41.38\%\\
TPPF(RE) & 91.42\% & 80.55\% & 55.87\% & 47.10\%\\
TPPF(CE) & 89.84\% & 74.64\% & 48.45\% & 41.59\%\\
TPPF(RECE) & 90.85\% & 79.11\% & 55.91\% & 47.09\%\\
FA-APF & 90.10\% & 76.29\% & 51.22\% & 44.24\%\\
iAPF & 100\% & 100\% & 100\% & 100\%\\
\hline
\end{tabular}
\caption{ESS}
\end{subtable}\hfill
\begin{subtable}{0.60\linewidth}
\centering
\begin{tabular}{ccccc}
\hline
  &d=2 & d=5 & d=15 &d=20  \\
\hline
BPF&2.93& 18.45 &6.52e3 &8.40e4\\
TPPF(RE)&2.43& 11.78& 1.48e3&1.62e4\\
TPPF(CE)&1.41& 4.89&0.14e3 &0.27e4\\
TPPF(RECE)&2.10& 7.22&7.10e3 &8.58e4\\
FA-APF&1.40& 5.19&1.60e3 &2.96e4\\
iAPF&4.59e-8& 7.60e-8&4.34e-8 &8.75e-8\\
\hline
\end{tabular}
\caption{Relative variance}
\end{subtable}
\caption{\textbf{Linear Gaussian model:} compare TPPF (trained with $L_{RE}$, $L_{CE}$, or $L_{RECE}$) and its competitors (BPF, iAPF and FA-APF). (a): Averaged relative ESS with 1000 replicates for $d \in \{2,5,15,20 \}$. (b): Relative variance computed using Monte Carlo approximation with $10^6$ samples for $d \in \{2,5,15,20 \}$.}
\label{table:combined}
\end{table}


\subsection{NGM-78 model}
In contrast to the linear nature of the last model, here we consider an (artificial) nonlinear model frequently used when testing the performance of particle filters \cite{netto1978optimal,gordon1993novel,kitagawa1996monte,west1993mixture}. To our knowledge, this model was first used by Netto, Gimeno, and Mendes in 1978 \cite{netto1978optimal}, so here we name it NGM-78 model for simplicity. The NGM-78 model describes the following discrete-time Markov chain in $\mathbb{R}^d$:

\begin{equation*}
    X_n = a_0 X_{n-1} + a_1 \frac{X_{n-1}}{(1 + |X_{n-1}|^2)} + f(n) + v_n,
\end{equation*}
and the observations
\begin{equation*}
    Y_n = a_2  |X_n|^2 + u_n,
\end{equation*}
where $v_n \sim N(0,\sigma^2_v I_d)$ and $u_n \sim N(0, \sigma^2_u I_d)$. In our experiments, we choose $a_0 = \frac{1}{2}$, $a_1 = 25$, $a_2 = \frac{1}{20}$, $f \equiv 0$, $\sigma_u^2 = 1$, $\sigma_v^2 = 0.01$, and $n = 0,1,\dots, 50$ ($n = 0,1,\dots, 10$ when computing the relative variance). The neural network structure is chosen to be a 2-layer network with each width 10. The learning rate is set to be 0.01, the optimization method is chosen to be ADAM, and the particle number is set to be 200. We test our TPPF algorithm on the NGM-78 model with the configuration $d \in \{1,2,5,10,15,20 \}$ and compare TPPF with its competitors. The empirical standard deviation of $\log Z$ is reported in Table \ref{tab:NGM}, and the boxplot is reported in Figure \ref{fig:NGM}. The relative variance is reported in Table \ref{tab:NGMrv}. As we can see, in contrast to the linear Gaussian model, in the current nonlinear settings, the linear structure (twisted functions learned in some Gaussian function family) of iAPF leads to its relatively worse behavior, which is clearer if we observe the relative variance reported in Table \ref{tab:NGMrv}. On the other hand, our TPPF algorithms behave better than iAPF partly due to the stronger expressive ability of neural networks. Moreover, in this model the term $\hat{P}[g](k,x)$ (recall the definition in \eqref{eq:twistedtransition}) used in FA-APF is calculated via Monte Carlo approximation since there is no analytical solution for it. Consequently, we observe in Table \ref{tab:NGM} that although the FA-APF outperforms other algorithms when $d=1$, it suffers from the curse of dimensionality and we cannot obtain reasonable estimates with the FA-APF in a feasible computational time.


\begin{table}[H]
    \centering
\begin{tabular}{cccccc}
\toprule  
 &d=1 & d=2  &  d=5  &d=15 &d=20 \\ 
\midrule 
BPF & 0.0243& 0.119 & 0.156 &0.252 &0.295 \\
TPPF(RE) &0.0184 & 0.0516 &  0.103 & 0.208 &0.246\\
TPPF(CE) &0.0220 & 0.0982 &  0.178&0.183&0.273\\
TPPF(RECE) & 0.0283& 0.0482 & 0.114 &0.187&0.280\\
FA-APF &0.00380 & 0.0718 & 0.162 &-&-\\
iAPF &0.0372 & 0.115 &  0.244&0.335&0.353\\
\bottomrule     
\end{tabular}
    \caption{\textbf{NGM-78 model:} compare TPPF (trained with $L_{RE}$, $L_{CE}$, or $L_{RECE}$) and its competitors (BPF, iAPF and FA-APF). Empirical \textbf{standard deviation} of $\log Z$ with 20 replicates for $d \in \{1,2,5,15,20 \}$. For $d\geq 15$, we cannot obtain reasonable estimates with the FA-APF in a feasible computational time.}
    \label{tab:NGM}
\end{table}

\begin{figure}[htbp]
    \centering
    \begin{subfigure}{0.45\textwidth}
        \centering
        \includegraphics[width=\textwidth]{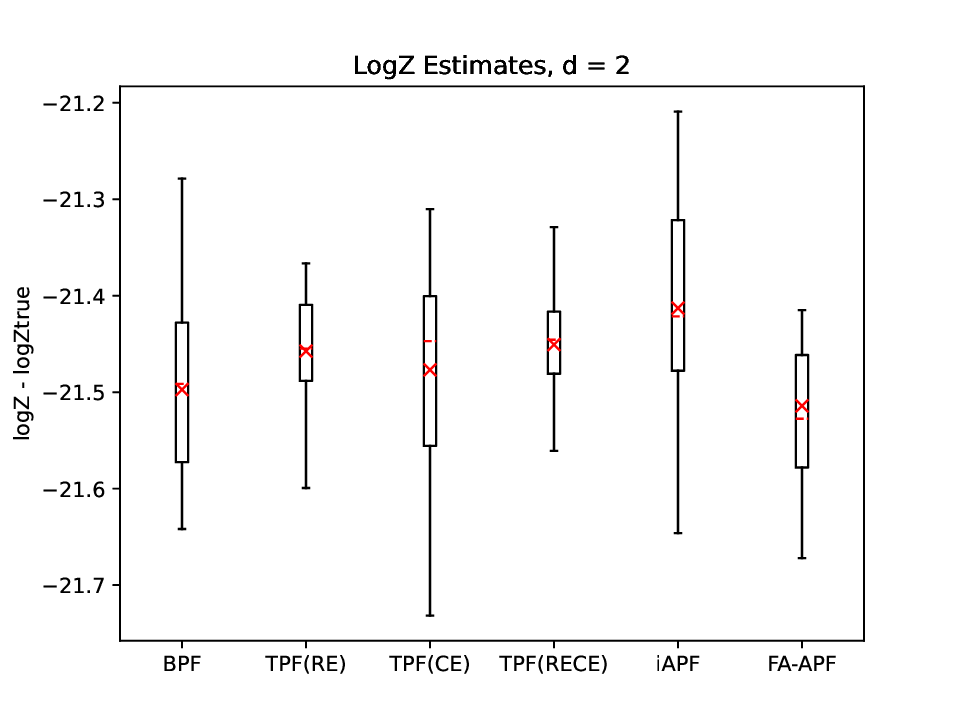}
        \caption{ }
        \label{fig:NGM2d}
    \end{subfigure}
    \hfill
    \begin{subfigure}{0.45\textwidth}
        \centering
        \includegraphics[width=\textwidth]{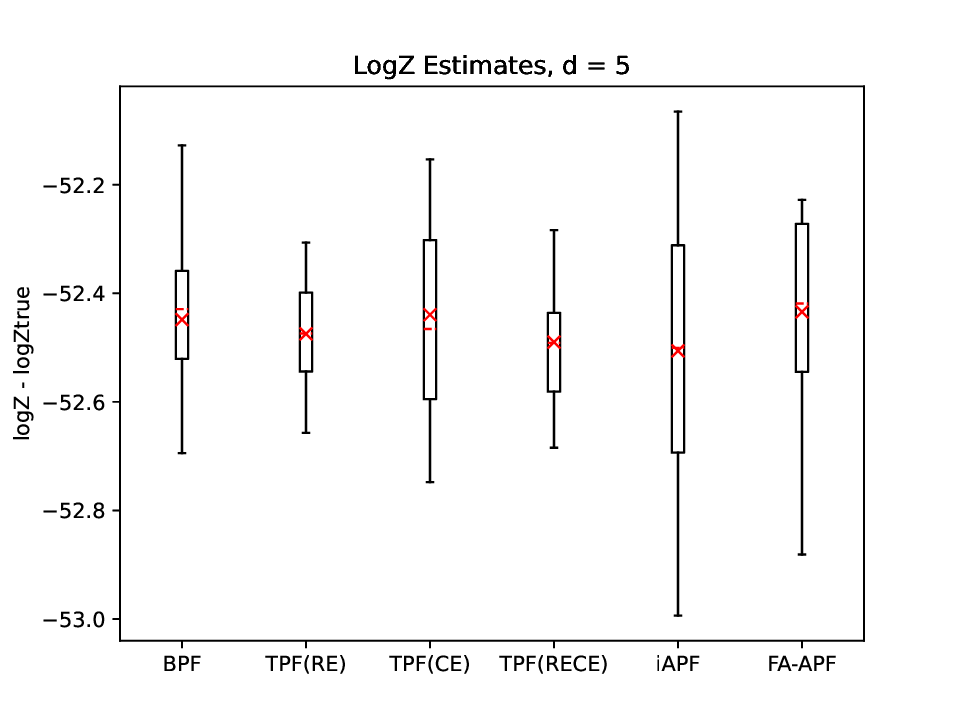}
        \caption{ }
        \label{fig:NGM5d}
    \end{subfigure}
    \caption{\textbf{NGM-78 model:} compare TPPF (trained with $L_{RE}$, $L_{CE}$, or $L_{RECE}$) and its competitors (BPF, iAPF and FA-APF).  \textbf{Boxplot of $\log Z$} using 20 replicates for (a) $d = 2$ and (b) $d=5$. The red cross represents the mean and the red dashed line represents the median. }
    \label{fig:NGM}
\end{figure}

\begin{table}[H]
    \centering
\begin{tabular}{cccccc}
\toprule  
 & d=1 & d=2 & d=5  & d=15 & d=20 \\ 
\midrule 
BPF & 0.26 & 0.39 & 0.43  & 0.55 & 0.53 \\
TPPF(RE) & 0.22 & 0.37 & 0.01  & 0.50 & 0.52 \\
TPPF(CE) & 0.25 & 0.33 & 0.39  & 0.61 & 0.48 \\
TPPF(RECE) & 0.24 & 0.36 & 0.44  & 0.49 & 0.52 \\
FA-APF & 0.03 & 0.26 & 0.34  & -- & -- \\
iAPF & 0.60 & 0.60 & 0.78 &   1.94 & 2.10 \\
\bottomrule     
\end{tabular}
\caption{\textbf{NGM-78 model:} compare TPPF (trained with $L_{RE}$, $L_{CE}$, or $L_{RECE}$) and its competitors (BPF, iAPF and FA-APF). \textbf{Relative variance} computed using Monte Carlo approximation with $10^6$ samples for $d \in \{1,2,5,15,20 \}$. For $d\geq 15$, we cannot obtain reasonable estimates with the FA-APF in a feasible computational time.}
    \label{tab:NGMrv}
\end{table}

\subsection{Lorenz-96 model}

We consider the Lorenz-96 (a nonlinear model) with white noise and partial observation \cite{lorenz1996predictability,heng2020controlled,murray2021anytime}. The discrete model can be viewed as the Euler-Maruyama scheme of the following interacting particle system consisting of $d$ particles in $\mathbb{R}^1$:
\begin{equation}\label{eq:SDELorenz96}
    dX^i = \left(-X^{i-1}X^{i-2} + X^{i-1}X^{i+1} - X^i + \alpha \right)dt + \sigma^2 dB^i,\quad 1 \leq i \leq d,
\end{equation}
where $\alpha \in \mathbb{R}$, $\sigma^2 \in \mathbb{R}^{+}$, $B^i$ ($1\leq i\leq d$) are independent Brownian motions, and the indices should be understood modulo $d$. The observation is through $Y_t \sim \mathcal{N}(\cdot; HX_t,\Sigma_{OB})$, where $H$ is a diagonal matrix with $H_{ii} = 1$ ($1 \leq i \leq d-2$), $H_{ii} = 0$ ($d-1 \leq i \leq d$). Note that this model does not have an analytical solution. We use the non-parametric implementation in this example.

Choose $\alpha = 3.0$, $\Sigma_{OB} = I_d$, $\sigma = 1$. The discrete-time Markov chain we study is the Euler-Maruyama scheme for \eqref{eq:SDELorenz96} with time step $\Delta t = 0.01$ and total time $T = 0.5$ (we choose a smaller $T = 0.1$ when computing the relative variance). Consider $d \in \{3,5,10 \}$. Similarly to the NGM-78 model, the neural network structure is chosen to be a 2-layer network, each with width 10. The learning rate is set to be 0.001, the optimization method is chosen to be ADAM, and the particle number is set to be 200. We report the empirical standard deviations and the relative variance in Table \ref{table:combined2} below:

\begin{table}[H]
\centering
\begin{subtable}{0.50\linewidth}
\centering
\begin{tabular}{cccc}
\hline
&d=3 & d=5  &  d=10  \\ 
\midrule 
BPF & 0.36& 0.64 &  2.14 \\
TPPF(RE) &0.26 & 0.43 &  1.23 \\
TPPF(CE) &0.31 & 0.48 &  1.73\\
TPPF(RECE) & 0.21& 0.51 & 1.69 \\
FA-APF &0.31 & 0.62 & 1.86 \\
iAPF &0.29 & 0.48 &  1.81\\
\hline
\end{tabular}
\caption{Standard deviation}
\end{subtable}\hfill
\begin{subtable}{0.50\linewidth}
\centering
\begin{tabular}{ccc}
\hline
 d=3 & d=5 & d=10 \\
\hline
0.71 & 1.38 & 5.04 \\
0.71 & 1.33 & 4.78 \\
0.70 & 1.34 & 4.66 \\
0.72 & 1.30 & 4.82 \\
0.61 & 1.16 & 3.46 \\
0.32 & 0.34 & 1.01 \\
\hline
\end{tabular}
\caption{Relative variance}
\end{subtable}
\caption{\textbf{Lorenz-96 model:} compare TPPF (trained with $L_{RE}$, $L_{CE}$, or $L_{RECE}$) and its competitors (BPF, iAPF and FA-APF). (a): Empirical \textbf{standard deviation} of $\log Z$ with 20 replicates for $d \in \{3,5,10\}$. (b): \textbf{Relative variance} computed using Monte Carlo approximation with $10^6$ samples for $d \in \{3,5,10 \}$.  }
\label{table:combined2}
\end{table}

To test the sensitivity of the proposed method, we fix $d = 3$, $\Sigma_{OB}=I_d$, $\sigma=1$ and run the experiment for different $\alpha$. See the results in Figure \ref{fig:alpha}.



\begin{figure}[htbp]
    \centering
    \begin{subfigure}{0.5\textwidth}
        \centering
        \includegraphics[width=\textwidth]{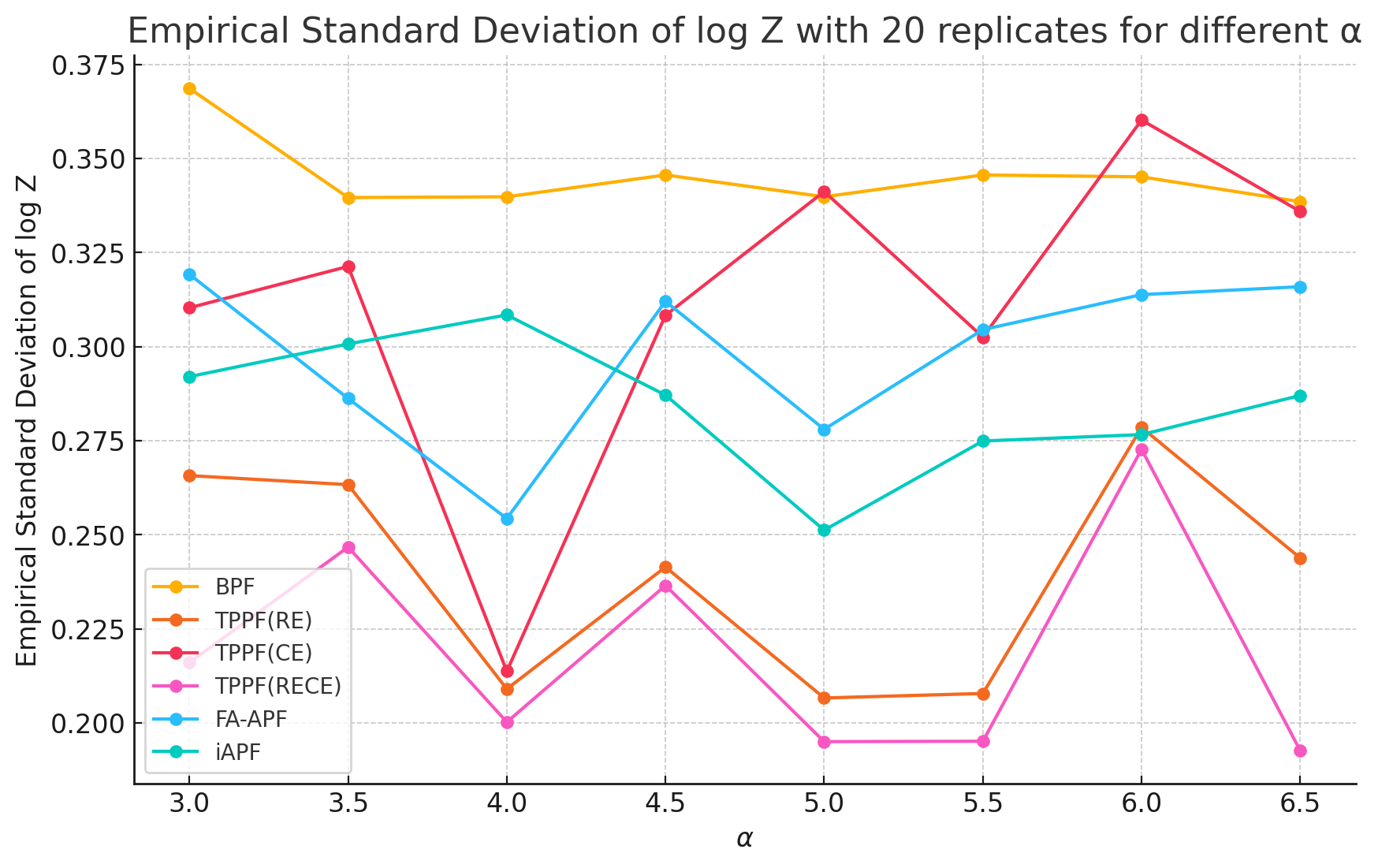}
        \caption{ }
        \label{fig:lorenz96l}
    \end{subfigure}
    \hfill
    \begin{subfigure}{0.45\textwidth}
        \centering
        \includegraphics[width=\textwidth]{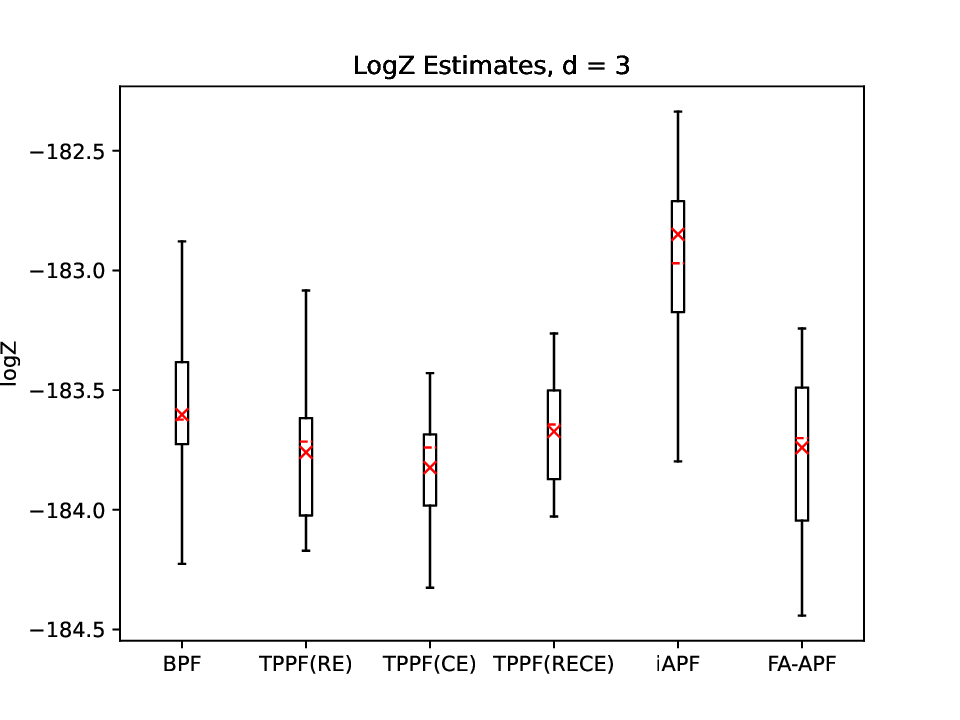}
        \caption{ }
        \label{fig:lorenz96r}
    \end{subfigure}
    \caption{\textbf{Lorenz-96 model:} compare TPPF (trained with $L_{RE}$, $L_{CE}$, or $L_{RECE}$) and its competitors (BPF, iAPF and FA-APF). \textbf{(a):} Empirical standard deviation for different external force strength $\alpha$ under dimension $d = 3$, using 20 replicates. \textbf{(b):} Boxplot of $\log Z$ using 20 replicates for $d = 3$, $\alpha = 3.0$. The red cross represents the mean and the red dashed line represents the median. }
    \label{fig:alpha}
\end{figure}


As we can see from Table \ref{table:combined2} (a) and Figure \ref{fig:alpha}, the proposed TPPF algorithm (especially the ones using loss functions $L_{RE}$ and $L_{RECE}$) behaves better than BPF and FA-APF, and slightly better than iAPF.  We also observe that for the current partially-observed nonlinear model, sometimes the result of iAPF is more biased (see (b) in Figure \ref{fig:alpha}), despite its relatively low variance.


\section{Conclusion}\label{sec:conclusion}
In this paper, we study the discrete-time twisted particle filters (TPF) from a continuous-time perspective. In detail, we propose a novel TPF algorithm, Twisted-Path Particle Filter (TPPF), inspired by algorithms in continuous-time settings. In TPPF, we choose some specific KL-divergence between path measures as the loss function and learn the twisting function parameterized by a neural network. Roughly speaking, this can be viewed as a policy gradient method under the current setting. We also give some numerical examples to illustrate the capability of the proposed algorithm. Some possible future work may include: seeking other practical TPF algorithms guided by the continuous-time importance sampling algorithms, such as other ideas from reinforcement learning. It is also interesting to try to understand existing TPF algorithms rigorously by finding their continuous-time limit.

\section*{Acknowledgments}
This work is supported in part by National Science Foundation via grant DMS-2309378. 
This work is done during Yuliang Wang's visit to Duke University. He thanks the Mathematics Department for their hospitality. 
The authors would like to thank Pieralberto Guarniero, Adam M. Johansen, Anthony Lee and Wei Deng for helpful discussions on implementation details of the iAPF method.

\appendix
\section{Proofs of Section \ref{sec:ISsetting}}\label{app:proofsec2}
\subsection{Proof of Proposition \ref{eq:discretegirsanov}}
\begin{proof}
Recall the notation $\hat{P}[\varphi](k,x) = \int \varphi(k,y') \hat{P}(x,dy')$. Then by definition,
\begin{equation*}
\begin{aligned}
&\quad\mathbb{E}_x\left[\prod_{k=0}^n\hat{g}_k^\varphi(\hat{X}^\varphi_k)\right]
= \int\hat{g}_0^\varphi(x) \prod_{k=1}^n  \hat{g}_k^\varphi(x_k) \hat{P}_k^\varphi(x_{k-1},x_k)dx_{1:n} \\
&=\int\hat{g}_0(x) \Big(\prod_{k=1}^n  \hat{g}_k(x_k)\Big)\Big(\frac{\hat{P}[\varphi](1,x)}{\varphi(n,x_n)}\prod_{k=1}^{n-1} \frac{\hat{P}[\varphi](k+1,x_k)}{\varphi(k,x_k)}  \prod_{k=1}^n \frac{\varphi(k,x_k)}{\hat{P}[\varphi](k,x_{k-1})} \hat{P}_k(x_{k-1},x_k)\Big)dx_{1:n} \\
&=\int\hat{g}_0(x) \prod_{k=1}^n  \hat{g}_k(x_k) \hat{P}_k(x_{k-1},x_k)dx_{1:n} = Z_{dis}(x).
\end{aligned}
\end{equation*}
\end{proof}

\subsection{Details on the optimal twisting function and the zero-variance property}

The optimal twisting function $\varphi^{*}$ has the following useful properties, which have also been discussed in related literature such as \cite{heng2020controlled,guarniero2017iterated}:
\begin{proposition}\label{prop:discretezerovariance}
Consider the functions $\varphi^{*}(k,\cdot)$ ($1\leq k \leq n$) defined by \eqref{eq:recursive}. Recall $\hat{X}_{0:n}$ is a Markov chain with transition density $\hat{P}(\cdot,\cdot)$. Then
\begin{enumerate}
    \item $\varphi^{*}(k,\cdot)$ satisfies
    \begin{equation}\label{eq:lookaheadapp}
        \varphi^{*}(k,x) = \mathbb{E}\left[\prod_{i=k}^n \hat{g}_{i}(\hat{X}_{i}) \mid \hat{X}_{k} = x\right], \quad 0 \leq k \leq n, \quad x \in \mathbb{R}^d.
    \end{equation}
    In particular, $Z_{dis}(x) = \varphi^{*}(0,x)$, $\forall x \in \mathbb{R}^d$.
    \item (zero variance property) 
    Consider the Markov chain $\hat{X}^{\varphi^{*}}_{0:n}$ with the initial state $\hat{X}^{\varphi^{*}}_0 = x$ and the transition density
    \begin{equation}
    \hat{P}^{\varphi^{*}}_k(x,y) := \frac{\varphi^{*}(k,y) }{\int \varphi^{*}(k,y') \hat{P}(x,y')dy'}\hat{P}(x,y),\quad 1\leq k\leq n,
    \end{equation}
    Define
    \begin{equation}
    \hat{W}(\hat{X}^{\varphi^{*}}) := \hat{L}(\hat{X}^{\varphi^{*}}) \prod_{k=0}^n \hat{g}_k(\hat{X}_k^{\varphi^{*}})
    \end{equation}
    with $\hat{L}(\hat{X}^{\varphi^{*}}) := \prod_{k=0}^n\ell_k^{*}(\hat{X}^{\varphi^{*}}_k)$ and
    \begin{equation}
    \begin{aligned}
    &\ell_0^{\varphi^{*}}(x) := \int \varphi^{*}(1,y)\hat{P}(x,dy),\quad \ell_n^{\varphi^{*}}(x) := \frac{1}{\varphi^{*}(n,x)},\\
    &\ell^{\varphi^{*}}_k(x) := \frac{\int \varphi^{*}(k+1,y)\hat{P}(x,dy)}{\varphi^{*}(k,x)},\quad 1\leq k\leq n.
    \end{aligned}
    \end{equation}
    Then, $\hat{W}(\hat{X}^{\varphi^{*}})$ is an unbiased estimate of $Z_{dis}(x)$  with zero variance. Namely,
    \begin{equation}\label{eq:hatWequalZdis}
    \hat{W}(\hat{X}^{\varphi^{*}}) = \varphi^{*}(0,x) = Z_{dis}(x), \quad \forall x \in \mathbb{R}^d.
    \end{equation}
\end{enumerate}
\end{proposition}

\begin{proof}
The proof is straightforward.
\begin{enumerate}[wide]
    \item We first prove the first claim. Since $\varphi^{*}(n,x) = \hat{g}_n(x)$, direct calculation yields
\begin{equation*}
    \begin{aligned}
        &\varphi^{*}(n-1,x) = \hat{g}_{n-1}(x) \int \hat{g}_n(y) \hat{P}(x,dy) = \mathbb{E}\left[\hat{g}_n(\hat{X}_{n})\hat{g}_{n-1}(\hat{X}_{n-1})\mid \hat{X}_{n-1} = x\right],\\
        &\varphi^{*}(n-2,x) = \hat{g}_{n-2}(x) \int \mathbb{E}\left[\prod_{i=n-1}^n \hat{g}_i(\hat{X}_{i})\mid \hat{X}_{n-1} = y\right] \hat{P}(x,dy)\\
        &= \mathbb{E}\left[\prod_{i=n-2}^n \hat{g}_i(\hat{X}_{i})\mid \hat{X}_{n-2} = x\right],\\
        &\cdots\\
        &\varphi^{*}(k,x) = \hat{g}_{k}(x) \int \mathbb{E}\left[\prod_{i=k+1}^n \hat{g}_i(\hat{X}_{i})\mid \hat{X}_{k+1} = y\right] \hat{P}(x,dy)\\
        &= \mathbb{E}\left[\prod_{i=k}^n \hat{g}_i(\hat{X}_{i})\mid \hat{X}_{k} = x\right],\\
        &\cdots\\
        &\varphi^{*}(0,x) = \hat{g}_{0}(x) \int \mathbb{E}\left[\prod_{i=1}^n \hat{g}_i(\hat{X}_{i})\mid \hat{X}_1 = y\right] \hat{P}(x,dy)\\
        &= \mathbb{E}\left[\prod_{i=0}^n \hat{g}_i(\hat{X}_{i})\mid \hat{X}_{0} = x\right] = Z_{dis}(x).\\
    \end{aligned}
\end{equation*}
\item Next, we prove the unbiased and zero-variance property of $\hat{W}$. The unbiased property is a direct result of Proposition \ref{eq:discretegirsanov} and Remark \ref{rmk:discretegirsanov}, and is valid for any twisting function sequence $\varphi(k,\cdot)$ ($1\leq k \leq n$). To prove the zero-variance property, we observe that under the probability measure $P^{*}$,
\begin{equation*}
    \hat{W} = \prod_{k=0}^n\ell_k^{*}(\hat{X}_k) \prod_{k=0}^n \hat{g}_{k} (\hat{X}_k) = \frac{\hat{g}_n(\hat{X}_n)\varphi^{*}(0,\hat{X}_0)}{\varphi^{*}(n,\hat{X}_n)} = \varphi^{*}(0,\hat{X}_0) = Z_{dis}(\hat{X}_0).
\end{equation*}
\end{enumerate}
\end{proof}

Moreover, a direct corollary of the zero-variance property in Proposition \ref{prop:discretezerovariance} is that, with the optimal twisting function, after Monte Carlo approximation with $N$ samples, the output $Z^{N,\varphi^{*}}_{dis}(x)$ of twisted particle filter (Algorithm \ref{alg:tpf}) is a perfect approximation of the target $Z_{dis}(x)$.

\begin{corollary}\label{coro:zerovarMC}
Consider the optimal twisting function defined in \eqref{eq:recursive}. Then for any positive integer $N$ and $x \in \mathbb{R}^d$,
\begin{equation}
    Z^{N,\varphi^{*}}_{dis}(x) = Z_{dis}(x).
\end{equation}
\end{corollary}

\begin{proof}
By definition,
\begin{equation*}
    Z^{N.\varphi}_{dis}(x) = \prod_{k=0}^n \frac{1}{N}\sum_{i=1}^N \hat{g}^{\varphi^{*}}_k(\zeta_k^i),
\end{equation*}
and
\begin{equation*}
    \hat{g}^{\varphi^{*}}_n(x)=\hat{g}_n(x) \ell_n^{\varphi^{*}}(x) = \frac{\hat{g}_n(x)}{\varphi^{*}(n,x)} \equiv 1,
\end{equation*}
\begin{equation*}
    \hat{g}^{\varphi^{*}}_k(x) = \hat{g}_k(x) \ell_k^{\varphi^{*}}(x) = \hat{g}_k(x) \frac{\hat{P}[\varphi^{*}](k+1,x)}{\varphi^{*}(k,x)} \equiv 1,\quad 1\leq k \leq n-1,
\end{equation*}
\begin{equation*}
    \hat{g}^{\varphi^{*}}_0(x) = \hat{g}_0(x) \ell_0^{\varphi^{*}}(x) = \hat{g}_0(x) \hat{P}[\varphi^{*}](1,x) = \varphi^{*}(0,x).
\end{equation*}
Therefore, 
\begin{equation*}
    Z^{N,\varphi^{*}}_{dis}(x) =\varphi^{*}(0,x)= Z_{dis}(x),\quad \forall x \in \mathbb{R}^d.
\end{equation*}

\end{proof}

\subsection{Details on the optimal control and the zero-variance property}
Here we give more details on the optimal control and the zero-variance property mentioned in Section \ref{sec:consetting}.

\begin{proposition}\label{prop:changeofmeasurecon}
Fix $T>0$ and $x\in \mathbb{R}^d$. Given functions $h: \mathbb{R}_{+}\times \mathbb{R}^d \rightarrow \mathbb{R}$, $g: \mathbb{R}^d \rightarrow \mathbb{R}$, $b: \mathbb{R}^d \rightarrow \mathbb{R}^d$,
\begin{enumerate}
\item (Feynman-Kac representation) The solution to the PDE \eqref{eq:optimalu} has the following representation:
\begin{equation}\label{eq:representu}
        v^{*}(t,x) = \mathbb{E}\left[e^{\int^{T}_{t} h(s,X_s) ds} g(X_T) \mid X_{t} = x \right],\quad 0\leq t\leq T,\quad x\in\mathbb{R}^d,
\end{equation}
where the process $X_t$ solves the SDE \eqref{eq:SDE}. In particular, $Z_{con}(x) = v^{*}(0,x)$, $\forall x \in \mathbb{R}^d$.

\item (zero-variance property) Let $P$ be the probability measure such that under which $(B_t)_{t\geq 0}$ is a Brownian motion in $\mathbb{R}^d$. Then there exists a probability measure $P^{*}$ such that under $P^{*}$, the law of the controlled process with control $u^{*} = \sqrt{2}\partial_x \log v^{*}$
\begin{equation*}
    X^{u^{*}}_t = X_0^{u^{*}} + \int_0^t \left(b(X^{u^{*}}_s) + 2\partial_x \log v^{*}(s,X^{u^{*}}_s)\right)ds + \sqrt{2}B_t,\quad X^{u^{*}}_0 = x
\end{equation*}
is the same of the law of $X_t$ under $P$ (recall that $X_t$ satisfies \eqref{eq:SDE}). Moreover, define
\begin{equation*}
    L(t) := \exp\left(-\int_0^t \partial_x \log v^{*}(s,X^{u^{*}}_s) dB_s -  \frac{1}{2}\int_0^t \left|\partial_x \log v^{*}(s,X^{u^{*}}_s)\right|^2 ds\right),\quad 0\leq t \leq T,
\end{equation*}
and
\begin{equation*}
    W := e^{\int_0^T h(s,X^{u^{*}}_s) ds}g(X^{u^{*}}_T)L(T).
\end{equation*}
We have 
\begin{equation}\label{eq:equalascon}
    W = v^{*}(0,x) = Z_{con}(x) \quad P^{*}-a.s..
\end{equation}
\end{enumerate}
\end{proposition}
\begin{proof}
\begin{enumerate}[wide]
\item We first prove the Feynman-Kac representation formula. Fix $0\leq t\leq T$, define the following process for $s \in [t,T]$
\begin{equation*}
    Y(s) := e^{\int_{t}^s h(\tau,X_{\tau}) d\tau} v^{*}(s,X_s).
\end{equation*}
Clearly,
\begin{equation*}
    Y(t) = v^{*}(t,X_{t}),\quad Y(T) = e^{\int_{t}^T h(\tau ,X_{\tau}) d\tau} g(X_T).
\end{equation*}
By It\^o's formula,
\begin{equation*}
\begin{aligned}
    dY(s) &= e^{\int_{t}^s h(\tau,X_{\tau}) d\tau} h(s,X_s) v^{*}(s,X_s) ds\\
    &\quad+ e^{\int_{t}^s h(\tau,X_{\tau}) d\tau} \nabla_x v^{*}(s, X_s) \cdot \left(b(X_s) ds + \sqrt{2} dB_s \right)\\
    &\quad+ e^{\int_{t}^s h(\tau,X_{\tau}) d\tau} \Delta_x v^{*}(s,X_s)ds\\
    &\quad + e^{\int_{t}^s h(\tau,X_{\tau}) d\tau} \partial_t v^{*}(s,X_s) ds 
\end{aligned}
\end{equation*}
By \eqref{eq:optimalu}, we have
\begin{equation*}
    dY(s) = \sqrt{2} e^{\int_{t}^s h(\tau,X_{\tau}) d\tau} \nabla_x v^{*}(s, X_s) \cdot dB_s.
\end{equation*}
Hence, $Y(s)$ is a martingale, and consequently,
\begin{equation*}
    \mathbb{E}\left[Y(T)\mid X_{t} = x\right] = \mathbb{E}\left[Y(t)\mid X_{t} = x\right],
\end{equation*}
namely,
\begin{equation*}
    v^{*}(t,x) = \mathbb{E}\left[ e^{\int_{t}^T h(\tau, X_{\tau}) d\tau} g(X_T)\mid X_{t} = x\right],\quad \forall x \in \mathbb{R}^d,\quad \forall t\in[0,T].
\end{equation*}

\item The existence and expression of the Radon-Nikodym derivative $L(T)$ is guaranteed by the classical Girsanov's theorem. We focus on the derivation of the zero-variance property \eqref{eq:equalascon} here. 

For $s \in [0,T]$, define the process
\begin{equation*}
    \omega_s := e^{\int_0^s h(\tau, X^{u^{*}}_{\tau})d\tau}v^{*}(s,X^{u^{*}}_s)L(s).
\end{equation*}
Clearly, 
\begin{equation*}
    \omega_0 = v^{*}(0,X_0),\quad \omega_T = W.
\end{equation*}
By It\^o's formula,
\begin{equation*}
\begin{aligned}
    d\omega_s &= h(s, X_s^{u^{*}})e^{\int_0^s h(\tau, X^{u^{*}}_{\tau})d\tau}v^{*}(s,X^{u^{*}}_s)L(s) ds\\
    &\quad+ e^{\int_0^s h(\tau, X^{u^{*}}_{\tau})d\tau}\partial_s v^{*}(s,X^{u^{*}}_s)L(s)ds\\
    &\quad+e^{\int_0^s h(\tau, X^{u^{*}}_{\tau})d\tau}\partial_x v^{*}(s,X^{u^{*}}_s) \cdot \left(b(X_s^{u^{*}}) ds + 2 \partial_x \log v^{*}(s,X^{u^{*}}_s) ds + \sqrt{2}dB_s \right)L(s)\\
    &\quad+e^{\int_0^s h(\tau,X^{u^{*}}_{\tau})d\tau}v^{*}(s,X^{u^{*}}_s)L(s)\left(-\sqrt{2}\partial_x \log v^{*}(s,X_s^{u^{*}}) dB_s \right)\\
    &\quad-e^{\int_0^s h(\tau, X^{u^{*}}_{\tau})d\tau}L(s)\left(\sqrt{2}\partial_x v^{*}(s,X_s^{u^{*}}) \cdot \sqrt{2} \partial_x \log v^{*}(s,X_s^{u^{*}})\right)ds.
\end{aligned}
\end{equation*}
Using the time evolution \eqref{eq:optimalu} for $u^{*}$, we have
\begin{equation*}
    d\omega_s \equiv 0.
\end{equation*}
Therefore
\begin{equation*}
    v^{*}(0,X_0) = \omega_0 = \omega_T = W \quad P^{*}-a.s..
\end{equation*}

\end{enumerate}
    
\end{proof}

\section{Proof of Section \ref{sec:TPPF}}
\subsection{Proof of Lemma \ref{lmm:DV}}
\begin{proof}
We refer to \cite{hartmann2017variational} for a similar proof. The proof is based on Jensen's inequality. In fact, by direct calculation and convexity of $-\log(\cdot)$, we have
\begin{multline*}
    -\log\mathbb{E}^{X \sim P}\left[\exp\left(-W(X) \right)\right] = -\log \int e^{-W(x)} P(dx) = -\log \int e^{-W(x)}\frac{dP}{dQ}(x) Q(dx)\\
    \leq \int \left(W(x) - \log \frac{dP}{dQ}(x)\right) Q(dx) =  \mathbb{E}^{X\sim Q}\left[W(X)\right] + \DKL(Q \| P), 
\end{multline*}
and the equality holds if and only if
\begin{equation*}
    \frac{dQ}{dP}(X) = \exp\left(-\log \mathbb{E}^{X \sim P}\left[\exp(-W(X))\right] -W(X)\right) .
\end{equation*}
\end{proof}

\subsection{Proof of Proposition \ref{eq:jvarphiexpression}}
\begin{proof}
By definition \eqref{eq:Jvarphi}, we have
\begin{equation*}
\begin{aligned}
    J(\varphi) &= \mathbb{E}_x\left[-\sum_{k=0}^n\log \hat{g}_k(\hat{X}_k)\right]\\
    &\quad+ \int \log \left(\prod_{k=1}^n\frac{\varphi(k,x_k)\hat{P}(x_{k-1},x_k)}{ \hat{P}[\varphi](k,x_{k-1})} / \prod_{k=1}^n \hat{P}(x_{k-1},x_k)\right) \prod_{k=1}^n\hat{P}^\varphi_k(x_{k-1},x_k) dx_{1:n}\\
    &= \mathbb{E}_x\left[-\sum_{k=0}^n \log \hat{g}_k(\hat{X}_k^{\varphi}) + \sum_{k=1}^n \log \frac{\varphi(k,\hat{X}_k^{\varphi})}{ \hat{P}[\varphi](k,X^{\varphi}_{k-1})}\right],
\end{aligned}
\end{equation*}
and
\begin{equation*}
\begin{aligned}
    &\quad \DKL(P^{\varphi}\|P^{\varphi^{*}}) = \int\log \left(\prod_{k=1}^n\frac{\varphi(k,x_k)}{ \hat{P}[\varphi](k,x_{k-1})} / \prod_{k=1}^n\frac{\varphi^{*}(k,x_k)}{ \hat{P}[\varphi^{*}](k,x_{k-1})}\right)\prod_{k=1}^n\hat{P}_k^\varphi(x_{k-1},x_k) dx_{1:n}\\
    & =\mathbb{E}_x\Big[\sum_{k=1}^n \log \frac{\varphi(k,\hat{X}^\varphi_k)}{ \hat{P}[\varphi](k,\hat{X}^{\varphi}_{k-1})}\Big]\\
    &\quad- \int\log \Big(\frac{\varphi^{*}(n,x_n)}{\varphi^{*}(0,x)}\prod_{k=0}^{n-1}\frac{\varphi^{*}(k,x_k)}{\hat{P}[\varphi^{*}](k+1,x_k)}\Big)\prod_{k=1}^n\hat{P}^\varphi_k(x_{k-1},x_k) dx_{1:n}\\
    &= \mathbb{E}_x\Big[-\sum_{k=0}^n \log \hat{g}_k(\hat{X}_k^{\varphi}) + \sum_{k=1}^n \log \frac{\varphi(k,\hat{X}_k^{\varphi})}{ \hat{P}[\varphi](k,\hat{X}^{\varphi}_{k-1})}\Big] + \log \varphi^{*}(0,x),
\end{aligned}
\end{equation*}
where we have used the recursive relation \eqref{eq:recursive} in the last equality. Moreover, using the derived expression for $J(\varphi)$ and \eqref{eq:recursive}, we know that
\begin{equation*}
    J(\varphi^{*}) = \log \varphi^{*}(0,x).
\end{equation*}
Consequently,
\begin{equation*}
    J(\varphi) = J(\varphi^{*}) + \DKL(P^{\varphi}\|P^{\varphi^{*}}).
\end{equation*}
\end{proof}

\subsection{Proof of Proposition \ref{prop:KLcontrolinequality}}
\begin{proof}
The proof relies on the following auxiliary results:
\begin{itemize}
    \item (Lemma \ref{lmm:generalJensen}, generalized Jensen's inequality) Given deterministic functions $\phi: \Omega \rightarrow \mathbb{R}$ and $f: \mathbb{R} \rightarrow \mathbb{R}$. Assume that $f$ is convex. Let $\lambda$, $\lambda'$ be two probability measures on $\Omega$ that are absolutely continuous with each other. Define the functional
    \begin{equation*}
        \mathcal{J}(f,\lambda,\phi) := \mathbb{E}_\lambda f(\phi) - f\mathbb{E}_\lambda \phi.
    \end{equation*}
    Then,
    \begin{equation}\label{eq:generalJensen}
        m\mathcal{J}(f,\lambda, \phi) \leq \mathcal{J}(f,\lambda', \phi) \leq M\mathcal{J}(f,\lambda, \phi),
    \end{equation}
    where $m := \inf_{E\in \mathcal{B}(\Omega)} \frac{\lambda'(E)}{\lambda(E)}$, $M := \sup_{E\in \mathcal{B}(\Omega)} \frac{\lambda'(E)}{\lambda(E)}$.
    \item (Lemma \ref{lmm:r2}, equivalence with the $\mathcal{X}^2$-divergence)
    \begin{equation}\label{eq:equivX2}
        r^2(Q) = \mathcal{X}^2(Q^{*} |Q),
    \end{equation}
    where the $\mathcal{X}^2$-divergence is defined by
    \begin{equation*}
        \mathcal{X}^2(Q^{*} | Q):=\mathbb{E}_Q\left[\left|\frac{dQ^{*}}{dQ}\right|^2 - 1\right].
    \end{equation*}
\end{itemize}
Now, by Jensen's inequality, we have
\begin{equation}\label{eq:KLjensen}
    \DKL(Q^{*} \| Q) = \mathbb{E}_{Q^{*}}\left[\log \frac{dQ^{*}}{dQ}\right] \leq \log\mathbb{E}_{Q^{*}}\left[ \frac{dQ^{*}}{dQ}\right]
\end{equation}
Combining \eqref{eq:KLjensen} and \eqref{eq:equivX2}, we have
\begin{equation*}
    r^2(Q) = \mathcal{X}^2(Q^{*} | Q) = \mathbb{E}_Q\left[\left|\frac{dQ^{*}}{dQ}\right|^2 - 1\right] = \mathbb{E}_{Q^{*}}\left[\frac{dQ^{*}}{dQ}\right] - 1 \geq e^{\DKL(Q^{*} \| Q)} - 1.
\end{equation*}
For the second claim in Proposition \ref{prop:KLcontrolinequality}, we choose $\lambda' = Q^{*}$, $\lambda = Q$, $\phi = \frac{dQ^{*}}{dQ}$, and $f = -\log$ in \eqref{eq:generalJensen}. Then,
\begin{equation}\label{eq:mathcalJ1}
    \mathcal{J}(f,\lambda,\phi) = -\mathbb{E}_{Q}\log \frac{dQ^{*}}{dQ} + \log \mathbb{E}_{Q} \frac{dQ^{*}}{dQ} = \DKL(Q \| Q^{*}),
\end{equation}
and
\begin{equation}\label{eq:mathcalJ2}
\begin{aligned}
    \mathcal{J}(f,\lambda',\phi) &= -\mathbb{E}_{Q^{*}} \log \frac{dQ^{*}}{dQ} + \log \mathbb{E}_{Q^{*}} \frac{dQ^{*}}{dQ}\\
    &= -\DKL(Q^{*} \| Q) + \log (\mathcal{X}^2(Q^{*} | Q) + 1) = -\DKL(Q^{*} \| Q) + \log (r^2(Q) + 1),
\end{aligned}
\end{equation}
where we have used \eqref{eq:equivX2} in the last equality. Combining \eqref{eq:generalJensen}, \eqref{eq:mathcalJ1} and \eqref{eq:mathcalJ2}, we obtain the second claim 
\begin{equation*}
        e^{m\DKL(Q \| Q^{*}) + \DKL(Q^{*} \| Q)}-1 \leq r^2(Q) \leq e^{M\DKL(Q \| Q^{*}) + \DKL(Q^{*} \| Q)}-1,
\end{equation*}
where $m:=\inf_E \frac{Q^{*}(E)}{Q(E)}$ and $M := \sup_E \frac{Q^{*}(E)}{Q(E)}$.
\end{proof}

The two auxiliary lemmas used in the proof of Proposition \ref{prop:KLcontrolinequality} are given below:

\begin{lemma}[generalized Jensen's inequality]\label{lmm:generalJensen}
Given deterministic functions $\phi: \Omega \rightarrow \mathbb{R}$ and $f: \mathbb{R} \rightarrow \mathbb{R}$. Assume that $f$ is convex. Let $\lambda$, $\lambda'$ be two probability measures on $\Omega$ that are absolutely continuous with each other. Define the functional
    \begin{equation*}
        \mathcal{J}(f,\lambda,\phi) := \mathbb{E}_\lambda f(\phi) - f(\mathbb{E}_\lambda \phi).
    \end{equation*}
    Then,
    \begin{equation}
        m\mathcal{J}(f,\lambda, \phi) \leq \mathcal{J}(f,\lambda', \phi) \leq M\mathcal{J}(f,\lambda, \phi),
    \end{equation}
    where $m := \inf_{E\in \mathcal{B}(\Omega)} \frac{\lambda'(E)}{\lambda(E)}$, $M := \sup_{E\in \mathcal{B}(\Omega)} \frac{\lambda'(E)}{\lambda(E)}$.
\end{lemma}

\begin{proof}
We first show that
\begin{equation}\label{eq:goaljenson}
    m\mathcal{J}(f,\lambda, \phi) \leq \mathcal{J}(f,\lambda',\phi).
\end{equation}
Taking $E = \Omega$, we have 
\begin{equation*}
    m := \inf_E \frac{\lambda'(E)}{\lambda(E)} \leq \frac{\lambda'(E)}{\lambda(E)} = 1.
\end{equation*}
Without loss of generality, assume $m<1$. (If $m=1$, then $\lambda \equiv \lambda'$, and the argument is trivial). \eqref{eq:goaljenson} is then equivalent to
\begin{equation}\label{goaljensenequiv}
    m\mathbb{E}_{\lambda} f(\phi) - mf(\mathbb{E}_\lambda \phi) \leq \mathbb{E}_{\lambda'}f(\phi) - f(\mathbb{E}_{\lambda'}\phi).
\end{equation}
Since $m \in (0,1)$, and $m = \inf_E \frac{\lambda'(E)}{\lambda(E)}$, the probability $\bar{\lambda} := \frac{\lambda' - m\lambda}{1-m}$ is well-defined. Then, by Jensen's inequality, since $f$ is convex, we have
\begin{equation*}
    \mathbb{E}_{\lambda'}f(\phi) - m\mathbb{E}_\lambda f(\phi) = (1-m)\mathbb{E}_{\bar{\lambda}} f(\phi) \geq (1-m)f(\mathbb{E}_{\bar{\lambda}} \phi) = (1-m)f\left(\frac{\mathbb{E}_{\lambda'}\phi - m\mathbb{E}_{\lambda} \phi}{1-m} \right).
\end{equation*}
Using convexity of $f$ again, we have
\begin{equation*}
    (1-m)f\left(\frac{\mathbb{E}_{\lambda'}\phi - m\mathbb{E}_{\lambda} \phi}{1-m} \right) + mf(\mathbb{E}_\lambda \phi) \geq f(\mathbb{E}_{\lambda'} \phi).
\end{equation*}
Therefore, \eqref{goaljensenequiv} holds, and thus
\begin{equation*}
    m\mathcal{J}(f,\lambda, \phi) \leq \mathcal{J}(f,\lambda',\phi).
\end{equation*}
Assuming $M>1$ and using exactly the same arguments, we have
\begin{equation*}
     M\mathcal{J}(f,\lambda, \phi) \geq \mathcal{J}(f,\lambda', \phi).
\end{equation*}
\end{proof}

\begin{lemma}\label{lmm:r2}
Recall the definitions of $Q$, $Q^{*}$ and the relative variance $r(Q)$ in Proposition \ref{prop:KLcontrolinequality}. Then
    \begin{equation}
        r^2(Q) = \mathcal{X}^2(Q^{*} | Q),
    \end{equation}
    where the $\mathcal{X}^2$-divergence is defined by
    \begin{equation*}
        \mathcal{X}^2(Q^{*} | Q):=\mathbb{E}_Q\left[\left|\frac{dQ^{*}}{dQ}\right|^2 - 1\right].
    \end{equation*}
\end{lemma}

\begin{proof}
By definition, we have
\begin{multline*}
    \mathcal{X}^2(Q^{*} | Q):=\mathbb{E}_Q\left[\left|\frac{dQ^{*}}{dQ}\right|^2 - 1\right] = \mathbb{E}_{Q}\left| \frac{dQ^{*}}{dQ}\right|^2 - \left|\mathbb{E}_{Q} \frac{dQ^{*}}{dQ}\right|^2\\ = \Var_{Q}\left(\frac{dQ^{*}}{dQ} \right)  =  \Var_{Q}\left(\frac{dQ^{*}}{dP} \frac{dP}{dQ} \right).
\end{multline*}
Recall the zero-variance property of $Q^{*}$:
\begin{equation*}
    \frac{dQ^{*}}{dP} = \frac{e^{-W}}{Z} \quad P-a.s. .
\end{equation*}
Consequently,
\begin{equation*}
    \Var_{Q}\left(\frac{dQ^{*}}{dP} \frac{dP}{dQ} \right) = \Var_{Q}\left(\frac{e^{-W}}{Z} \frac{dP}{dQ} \right) = \frac{\Var_{Q}\left(e^{-W}\frac{dP}{dQ} \right)}{Z^2} = r^2(Q).
\end{equation*}
Hence, 
\begin{equation*}
    r^2(Q) = \mathcal{X}^2(Q^{*} | Q).
\end{equation*}
\end{proof}

\section{Implementation details for numerical examples}\label{sec:details}

Here we give some details on how to parameterize the twisting function $\varphi$, how to sample from the twisted Markov transition kernel $\hat{P}^{\varphi}$, and how to calculate the normalizing constant $\hat{P}[\varphi](k,x)$. In fact, in the following three experiments, we consider two ways of parameterization for the twisting function: the robust non-parametric way and the problem-dependent parametric way. In particular, for the linear Gaussian model, we use the problem-dependent implementation, while for the other two models we use the non-parametric implementation.

\begin{enumerate}[wide]
\item The non-parametric implementation. As discussed in Section \ref{sec:TPPF}, in most applications, we approximate the twisting function by
\begin{equation*}
    \log \varphi(k,x) = NN(\theta;k,x),
\end{equation*}
where $NN(\theta;k,x)$ is a neural network with parameters $\theta$ and inputs $k$, $x$ ($0\leq k \leq n$, $x \in \mathbb{R}^d$). In our experiments, we set $NN$ as DenseNet \cite{iandola2014densenet,richter2021solving} with two hidden layers. Moreover, we add a tanh activation to the final layer so that the output $NN(\theta;k,x)$ takes value in $(\epsilon,1)$, where $\epsilon>0$ is a hyperparameter. 
This boundedness restriction is designed for the following rejection sampling step when sampling from the twisted kernel $\hat{P}^\varphi$. In fact, since we usually do not have much information about the current twisting function $\varphi$, it is not easy to sample from the twisting kernel $\hat{P}^\varphi_k(x,\cdot) \sim \varphi(k,\cdot)\hat{P}(x,\cdot)$. Under the non-parametric implementation setting, we make use of the rejection sampling recently proposed in \cite{bon2022monte}: \\
\textit{Using $X^{\varphi}_k$ and the untwisted transition kernel $\hat{P}(X^\varphi_k,\cdot)$, propose a new position $X_{pro}$, accept it with probability $\varphi(X_{pro})$.
Repeat until first acceptance.}

Moreover, in the non-parametric implementation setting, the normalizing constant $\hat{P}[\varphi]$ is calculated via Monte Carlo approximation using $\tilde{N}$ samples (following \cite{bon2022monte}, we choose $\tilde{N} = 50$ in both Lorenz-96 and NGM-78 models):
\begin{equation*}
    \hat{P}[\varphi](k,x) = \int \varphi(k,y) \hat{P}(x,dy) \approx \frac{1}{\tilde{N}}\sum_{i=1}^{\tilde{N}} \varphi(k,U_i),\quad U_i \sim \hat{P}(x,\cdot)\quad i.i.d.
\end{equation*}

Also, to make the training faster, we use the untwisted process $X$ instead of $X^{\varphi}$ when calculating the loss and the gradient (i.e. for $L_{RE}$, the loss is computed via the first line in \eqref{eq:LRE} and its gradient is computed using \eqref{eq:gradient}).

\item The parametric implementation. In the linear Gaussian model, we can calculate the analytical solution of $Z_{dis}(x)$ using the Kalman filter \cite{kalman1960new}. Moreover, we can analytically calculate the optimal twisting function using the backward recursive relation \eqref{eq:recursive}, and clearly the optimal twisting function is also Gaussian. Therefore, with so much knowledge of the solution, it is reasonable to consider a problem-dependent way to parameterize the twisting function to make the learning more efficient. In more details, via a mean-variance estimation framework, we set
\begin{equation*}
    \mu_k = NN_1(\theta_1;k) \in \mathbb{R}^d,\quad \sigma_k^2 = NN_2(\theta_2;k) \in \mathbb{R}_{+},\quad 0\leq k \leq n.
\end{equation*}
And then set
\begin{equation}\label{eq:MVE}
    \varphi(k,x) = C_k N(x;\mu_k,\sigma^2_k).
\end{equation}
Note that the twisting function is invariant of the constant scaling, so we do not care about the constant that multiplies the Gaussian. This means there is no need to learn $C_k$ in \eqref{eq:MVE}, and in our experiment we just consider $C_k = (2\pi \sigma_k^2)^{-\frac{d}{2}}$ so that $\varphi(k,x) = \exp(-|x-\mu_k|^2/2\sigma^2_k)$. Also, under such settings, we no longer need rejection sampling and the inner-loop Monte Carlo for calculating the normalizing constant, because everything can be calculated analytically. Consequently, compared with the non-parametric way, the parametric implementation is less time-consuming and the neural network is easier to train.
\end{enumerate}

Another important remark is about the relative fairness of the comparisons in the numerical experiments. We learn and run the particle filter using the same particle numbers. For the time complexity, we admit that as a training-based algorithm, our algorithm is more time-consuming, and in fact, there is not a completely fair comparison due to different non-optimal choices of network structures. In order to conduct a relatively fair comparison, we restrict the number of iterations for the training so that the total running time is comparable and similar to its competitors.





 







\section{A convergence analysis: From discrete-time to continuous-time models}\label{sec:convergence}
Here we provide a rigorous proof from the discrete-time model to the continuous-time model (recall their settings in Section \ref{sec:ISsetting}). We postpone some of the technical proofs to the end of this section.
In order to study the connection between the models, we need the following restriction for their transition kernels $P$, $\hat{P}$ and the functions determining the targets $Z_{con}(x)$, $Z_{dis}(x)$, so that we can establish the convergence results rigorously below. In detail, fix the function $b: \mathbb{R}^d \rightarrow \mathbb{R}^d$, time step $\eta > 0$, and $T :=n\eta$. We consider the Gaussian transition kernel for the discrete model:
\begin{equation}\label{eq:discretekernel}
    \hat{P}^\eta(x,dy) := (4\pi\eta)^{-\frac{d}{2}}\exp\left(-\frac{1}{4\eta}|y-x-\eta b(x)|^2 \right) dy.
\end{equation}
Recall that the transition kernel for the continuous model $P_t$ is associated with the SDE \eqref{eq:SDE} with drift $b(\cdot)$ and volatility $\sqrt{2}$, and it satisfies the Fokker-Planck equation \eqref{eq:FPcon}. Clearly, $\hat{P}^\eta$ is the one-step transition kernel of the corresponding Euler-Maruyama scheme. Moreover, for convergence analysis in this section, we consider the discrete model determined by $(\hat{P}^\eta; \hat{g}_{0:n}(\cdot); x)$ and the continuous model determined by $(P_t;\log g(\cdot,\cdot),g_T(\cdot);x)$. Correspondingly, the statistical quantities of interest are respectively
\begin{equation}\label{eq:ZdisSec3}
    Z_{dis}(x) := \mathbb{E}_x\left[\prod_{k=0}^{n-1} \hat{g}_k(\hat{X}_k)\hat{g}_n(\hat{X}_n)\right] = \mathbb{E}_x\left[\exp\left(\sum_{k=0}^{n-1}\int_{k\eta}^{(k+1)\eta} \log \left(\hat{g}_k(\hat{X}_k)\right)^{\eta^{-1}}ds\right)\hat{g}_n(\hat{X}_n)\right],
\end{equation}
and
\begin{equation}\label{eq:ZconSec3}
    Z_{con}(x) = \mathbb{E}_x\left[\exp\left(\int_0^{T} \log g(s,X_s) ds\right) g_T(X_{T})\right].
\end{equation}
In what follows, under suitable assumptions for the functions $b(\cdot)$, $\hat{g}_k(\cdot)$, $g(\cdot,\cdot)$ and $g_T(\cdot)$, we will consider the convergence from discrete-time model to the continuous-time model. In detail, we will show that at the time step $\eta \rightarrow 0$, $\hat{P}^\eta$ converges to $P_t$, and $Z_{dis}(x)$ converges to $Z_{con}(x)$.

We need the following assumptions to ensure convergence:
\begin{assumption}\label{ass1}
We assume the following conditions for $b:\mathbb{R}^d \rightarrow \mathbb{R}^d$, $\hat{g}_k: \mathbb{R}^d \rightarrow \mathbb{R}$ ($0\leq k \leq n$), $g:\mathbb{R}_{+} \times \mathbb{R}^d \rightarrow \mathbb{R}$ and $g_T: \mathbb{R}^d \rightarrow \mathbb{R}$:
\begin{itemize}
    \item[(a)] $b(\cdot)$ is $L_b$-Lipschitz (i.e. $|b(x) - b(y)|\leq L_b |x-y|$, $\forall x,y\in \mathbb{R}^d$), and $\sup_{0\leq i \leq n}\mathbb{E}|\hat{X}_{i\eta}|^2 < \infty$ uniformly in $\eta$.
    \item[(b)] $\hat{g}_n(x) \rightarrow g_T(x)$, $\eta^{-1}\log \hat{g}_{k}(x) \rightarrow \log g(k\eta, x)$ uniformly in $x$ and $k$ ($0 \leq k \leq n-1$) as $\eta \rightarrow 0$.
    \item[(c)] $\log g(t, x)$ is continuous in $t \in [0,T]$ uniformly in $x$. $\eta^{-1}\log \hat{g}_{k\eta}(x)$ is Lipschitz in $x$ uniformly for all $k$ ($0 \leq k \leq n-1$) and $\eta>0$.
    \item[(d)]  $\hat{g}_n(\cdot)$ is bounded. For any $t \in [0,T]$, the functional $X_{[t,T]} \mapsto e^{\int_t^T \log g(\tau, X_{\tau})d\tau}$ is continuous and bounded.
\end{itemize}
\end{assumption}
Above, condition (a) is used to ensure the existence and uniqueness of the strong solution for \eqref{eq:SDE} and the convergence of transition kernel in Proposition \ref{prop:EMerror} below. Other conditions are required for the convergence of $Z_{dis}(x)$ to $Z_{con}(x)$ in Proposition \ref{prop:optimalconverge}  below. Also, we will give a concrete example in Example \ref{example} below which satisfies all the conditions in Assumption \ref{ass1}.

We also remark here that the upper bound for the second moment $\sup_{0\leq i \leq n} \mathbb{E}|\hat{X}_{i\eta}|^2$ can be proved if one assumes: (1) $L_b$-Lipschitz condition for $b(\cdot)$; (2) the following confining condition for $b(\cdot)$:
\begin{equation*}
    x \cdot b(x) \leq -C_1 |x|^2 + C_2, \quad \forall x \in \mathbb{R}^d,
\end{equation*}
where $C_1$, $C_2$ are two positive constants. Moreover, under such assumptions, the upper bound $\sup_{0\leq i \leq n} \mathbb{E}|\hat{X}_{i\eta}|^2$ can be shown to be independent of the time $T$ (recall that $T = n\eta$). 

In the following proposition, we establish the convergence of transition kernels in terms of the KL-divergence or the total variation (TV) distance. Note that for two probability measures $\mu$, $\nu$ on $\mathbb{R}^d$, the KL-divergence and TV distance are given by
\begin{equation}
\DKL(\mu\|\nu) :=
\left\{
\begin{aligned}
     &\int_{\mathbb{R}^d}  \log \frac{d\mu}{d\nu} \, \mu(dx),\quad \mu \ll \nu,\\
     &+\infty,\quad \text{otherwise}.
\end{aligned}
\right.
\end{equation}
\begin{equation}
    \text{TV}(\mu,\nu) := \sup_{A\in \mathcal{B}(\mathbb{R}^d)} |\mu(A) - \nu(A)|. 
\end{equation}

\begin{proposition}\label{prop:EMerror}
Consider the Markov transition kernels $\hat{P}^{\eta}$ and $P_\eta (=P_{t=\eta})$ defined in \eqref{eq:discretekernel}, \eqref{eq:FPcon}, respectively. Suppose that the condition (a) in Assumption \ref{ass1} holds. Then for any $x\in\mathbb{R}^d$, as $\eta \rightarrow 0$, $\hat{P}^{\eta}(x,\cdot) \rightarrow P_{\eta}(x,\cdot)$ in terms of KL-divergence or TV distance. In detail, for $\eta < 1$, there exists a positive constant $C$ independent of $\eta$ and $d$ such that
\begin{equation}\label{eq:rateKL}
    \DKL\left( \hat{P}^{\eta}(x,\cdot)  \| P_{\eta}(x,\cdot)   \right) \leq Cd\eta^2 \rightarrow 0 \quad \text{as} \quad \eta \rightarrow 0,
\end{equation}
and
\begin{equation}
    \mathrm{TV}\left( \hat{P}^{\eta}(x,\cdot)  , P_{\eta}(x,\cdot)   \right) \leq C\sqrt{d}\eta \rightarrow 0 \quad \text{as} \quad \eta \rightarrow 0.
\end{equation}
\end{proposition}

Note that the proposition considers the local truncation error of the Euler-Maruyama scheme in terms of TV distance and KL divergence, consistent with the existing results \cite{dalalyan2017theoretical}: consider the solution $X_t$ ($t \geq 0$) to the continuous-time SDE \eqref{eq:SDE} and its Euler-Maruyama discretization $\hat{X}_{n\eta}$ ($n \in \mathbb{N}_{+}$) with time step $\eta$, for any time interval length $T$ such that $\eta$ divides $T$, one can show that 
\begin{equation*}
    D_{KL}\left(\text{Law}(\hat{X}_T)\| \text{Law}(X_T)\right) \lesssim T\eta.
\end{equation*}
This then reduces to \eqref{eq:rateKL} when $T = \eta$. Moreover, recent literature (for instance, \cite{mou2022improved,li2022sharp}) has shown that, if one assumes stronger smoothness conditions for the drift $b(\cdot)$, it is possible to improve the upper bound for the time-discretization error in terms of the KL-divergence from $O(T\eta)$ to $O(T\eta^2)$. Consequently, the rate in \eqref{eq:rateKL} is $O(\eta^3)$. The same arguments also hold for TV distance via Pinsker's inequality.

The next proposition guarantees that under the settings of Assumption \ref{ass1}, the target $Z_{dis}(x)$ converges to $Z_{con}(x)$ as $\eta \rightarrow 0$.

\begin{proposition}\label{prop:optimalconverge}
For $x \in \mathbb{R}^d$, recall the definitions of $Z_{dis}(x)$, $Z_{con}(x)$ in \eqref{eq:ZdisSec3}, \eqref{eq:ZconSec3}, respectively. Then under Assumption \ref{ass1},
\begin{equation}
    Z_{dis}(x) \rightarrow Z_{con}(x)
\end{equation}
pointwise as $\eta \rightarrow 0$.
\end{proposition}

\begin{remark}[explicit convergence rate]
Under current assumptions, one cannot obtain an explicit convergence rate due to conditions (b) and (c) in Assumption \ref{ass1}. However, if we use a stronger version of Assumptions \ref{ass1}, it is possible to obtain an explicit rate. In fact, if we replace conditions (b) and (c) in Assumption \ref{ass1} by the following stronger quantitative version 
\begin{itemize}
    \item[(b')] $g_T(\cdot)$ is bounded. $\left|\hat{g}_n(x) - g_T(x) \right| \leq C_1\eta^{\alpha_1}$, $\left|\eta^{-1}\log \hat{g}_{k\eta}(x) - \log g_{k\eta}(x) \right| \leq C_2\eta^{\alpha_2}$, $0 \leq k \leq n-1$, where $C_1$, $C_2$, $\alpha_1$, $\alpha_2$ are positive constants independent of $x$, $k$.
    \item[(c')]$\log g_s(x)$ is $\alpha_3$-H\"older continuous in $t \in [0,T]$ uniformly in $x$. $\eta^{-1}\log \hat{g}_{k\eta}(x)$ is Lipschitz in $x$ uniformly for all $k$.
\end{itemize}
then, following exactly the same derivation, one can easily obtain for small $\eta$,
\begin{equation}
    |Z_{dis}(x) - Z_{con}(x)| \leq \eta^{\alpha},
\end{equation}
where $\alpha := \min (\frac{1}{2}, \alpha_1, \alpha_2, \alpha_3)$.
\end{remark}

We end this section by giving a concrete example satisfying all the conditions in Assumption \ref{ass1}. In particular, it satisfies the time-continuity condition of $\log g(t,x)$, which might not be very direct at first glance.
\begin{example}[an example satisfying time-continuity and other assumptions]\label{example}
Note that one important assumption is that the function $\log g(t,x)$ is time-continuous. The following example shows that we can indeed find such a time-continuous $\log g(t,x)$.
In fact, we consider the following continuous-time state space model:
\begin{equation}\label{eq:examplenohat}
    X_t = x + \int_0^t b(X_s) ds + \sqrt{2}B_t, \quad Y_t = X_t + \tilde{B}_{t+t_0}, \quad 0\leq t\leq T
\end{equation}
for two independent Brownian motions $(B_t)_{t\geq 0}$, $(\tilde{B}_t)_{t\geq 0}$. Here we introduce the positive constant $t_0$ only to avoid a potential singularity of $1/t$ at $t=0$: As will be discussed below, the convergence would require the time continuity of $- \frac{|x - y_t|^2}{2(t+t_0)}$ in the closed interval $[0,T]$, thus the choice of some $t_0 > 0$.

Correspondingly, for the discrete model, for $0\leq k\leq n$, we define $\hat{Y}_k$ by
\begin{equation}\label{eq:examplehat}
    \hat{X}_{k+1} = \hat{X}_k + \eta b(\hat{X}_k)  + \sqrt{2}(B_{(k+1)\eta} - B_{k\eta}),\quad \hat{Y}_k= \hat{X}_k + \tilde{B}_{k\eta+t_0},\quad 0\leq k \leq n.
\end{equation}

Now suppose that we have been given observation $y_t$ ($0\leq t\leq T$) and $\hat{y}_k$ ($0\leq k \leq n$) (note that by our construction above, the observed process $y_t$ of the continuous model is time-continuous). Clearly, the expressions for $g(t,x)$ and $\hat{g}_k(x)$ are explicit and are of the Gaussian form. Therefore in this case, we are able to convert the conditions in Assumption \ref{ass1} for $g$ and $\hat{g}$ to assumptions for $y$ and $\hat{y}$. Basically, defining $\mathcal{N}(\cdot;\mu,\Sigma)$ to be the density of the $d$-dimensional Gaussian distribution with mean $\mu$ and covariance $\Sigma$, we need the following:
\begin{itemize}
    \item $\mathcal{N}(x;\hat{y}_n,T+t_0) \rightarrow \mathcal{N}(x;y_n,T+t_0)$, $|x - \hat{y}_k|^2 \rightarrow |x-y_k|^2$ uniformly in $x$, $k$ as $\eta \rightarrow 0$.
    \item For any $t \in [0,T]$, the functional $X_{[t,T]} \mapsto e^{\int_t^T  -\frac{|X_{\tau} - y_\tau|^2}{2(t+t_0)}d\tau}$ is bounded.
    \item $- \frac{|x - y_t|^2}{2(t+t_0)}$ is continuous in $t \in [0,T]$ uniformly in $x$. $- \frac{|x - \hat{y}_{k}|^2}{2(k\eta+t_0)}$ is Lipschitz in $x$ uniformly for all $k$.
\end{itemize}
It is then easy to see that these conditions generally require (1) time-continuity for the observation $y_t$, (2) the observation for the two models $y_{k\eta}$, $\hat{y}_k$ are very close. Clearly, since $\hat{X}$ in \eqref{eq:examplehat} is the Euler-Maruyama discretization of $X$ in \eqref{eq:examplenohat}, the continuous observed processes $y$, $\hat{y}$ stay close when the time step $\eta$ is small. 

Furthermore, once these requirements are met, by Proposition \ref{prop:optimalconverge}, as $\eta \rightarrow 0$, the rescaled target
\begin{equation*}
    \tilde{Z}_{dis}(x) := \mathbb{E}_x\left[\prod_{k=0}^{n-1} (\hat{g}_k(\hat{X}_k))^{\eta^{-1}} \hat{g}_n(\hat{X}_n)\right]
\end{equation*}
converges pointwise to
\begin{equation*}
    Z_{con}(x) = \mathbb{E}_x\left[\exp\left(\int_0^{T} \log g(s,X_s) ds\right) g_T(X_{T})\right].
\end{equation*}

\end{example}


\subsection{Proof of Proposition \ref{prop:EMerror}}
\begin{proof}
Consider the following SDEs with the same initial state:
\begin{equation}\label{eq:twoSDEs}
    \begin{aligned}
        &dX_t = b(X_t) dt + \sqrt{2}dB,\quad X_0 = x,\quad 0\leq t \leq T=n\eta,\\
        &d\hat{X}_t = b(\hat{X}_{k\eta})dt + \sqrt{2}dB,\quad \hat{X}_0 = x, \quad t\in[k\eta,(k+1)\eta),\quad 0\leq k\leq n-1.
    \end{aligned}
\end{equation}
Clearly, $P_{\eta}(x,\cdot) = \text{Law}(X_\eta)$ and $\hat{P}^\eta(x,\cdot) = \text{Law}(\hat{X}_\eta)$.
Denote by $P^x_{[0,t]}$, $\hat{P}_{[0,t]}^x$ the path measures with the same initial $x$ associated with the time interval $[0,t]$ for any fixed $t \leq T$. Then using the data processing inequality \cite{beaudry2011intuitive,li2023propagation} and Girsanov's theorem \cite{girsanov1960transforming,durrett2018stochastic}, it holds that
\begin{equation}\label{eq:aftergirsanov}
    \DKL\left(\text{Law}(\hat{X}_t) \| \text{Law}(X_t) \right) \leq \DKL\left(\hat{P}_{[0,t]}^x \| P_{[0,t]}^x \right) \leq \mathbb{E}\Big[\sum_{k=0}^{\lceil t/\eta \rceil} \int_{k\eta}^{(k+1)\eta} |b(\hat{X}_s) - b(\hat{X}_{k\eta})|^2 ds\Big].
\end{equation}
By condition (a) in Assumption \ref{ass1}, $b(\cdot)$ is $L_b$-Lipschitz, and the second  moment for $\hat{X}_{k\eta}$ has uniform bound, then for any $0\leq k \leq \lceil t/\eta \rceil$, one has 
\begin{equation}\label{eq:smallinterval}
\begin{aligned}
    \mathbb{E}\left|b(\hat{X}_t) - b(\hat{X}_{k\eta})\right|^2&\leq L_b^2\mathbb{E}\left|(t-k\eta)b(\hat{X}_{k\eta}) + \int_{k\eta}^t dB_s \right|^2\\
    &\leq 2L_b^2\left(2\eta^2\left(|b(0)|^2 + L_b^2 \sup_{0\leq i \leq n} \mathbb{E}|\hat{X}_{i\eta}|^2\right) + \eta d \right) \leq Cd\eta,
\end{aligned}
\end{equation}
where we need $\eta < 1$ and $C = C\left(L_b, b(0), \sup_{0\leq i \leq n} \mathbb{E}|\hat{X}_{i\eta}|^2 \right)$ is a positive constant.
Combining \eqref{eq:aftergirsanov} and \eqref{eq:smallinterval}, we know that
\begin{equation*}
     \DKL\left(\text{Law}(\hat{X}_t) \| \text{Law}(X_t) \right) \leq Cdt\eta.
\end{equation*}
And by Pinsker's inequality, we have
\begin{equation*}
    \text{TV}\left(\text{Law}(\hat{X}_t) , \text{Law}(X_t) \right) \leq C\sqrt{dt\eta}.
\end{equation*}
Finally, taking $t = \eta$, we obtain the desired result.
\end{proof}

\subsection{Proof of Proposition \ref{prop:optimalconverge}}
\begin{proof}
Fix $T>0$, $\eta>0$, $x\in \mathbb{R}^d$. Recall that $T = n\eta$, and
\begin{equation*}
    Z_{dis}(x) := \mathbb{E}_x\left[\prod_{k=0}^{n-1} \hat{g}_k(\hat{X}_k)\hat{g}_n(\hat{X}_n)\right] = \mathbb{E}_x\left[\exp\left(\sum_{k=0}^{n-1}\int_{k\eta}^{(k+1)\eta} \log \left(\hat{g}_k(\hat{X}_k)\right)^{\eta^{-1}}ds\right)\hat{g}_n(\hat{X}_n)\right],
\end{equation*}
\begin{equation*}
        Z_{con}(x) = \mathbb{E}_x\left[e^{\int^{T}_{0} \log g(s,X_s) ds} g_T(X_T)  \right].
\end{equation*}
Then, for
\begin{equation*}
    \begin{aligned}
        &dX_t = b(X_t) dt + \sqrt{2}dB,\quad X_0 = x,\quad 0\leq t \leq T=n\eta,\\
        &d\hat{X}_t = b(\hat{X}_{k\eta})dt + \sqrt{2}dB,\quad \hat{X}_0 = x, \quad t\in[k\eta,(k+1)\eta),\quad 0\leq k\leq n-1,
    \end{aligned}
\end{equation*}
we have
\begin{equation*}
\begin{aligned}
    &\quad Z_{con}(x) - Z_{dis}(x)\\
    &= \mathbb{E}_x\left[\exp\left(\int^{T}_{0} \log g(s,X_s) ds\right) g_T(X_T) \right] - \mathbb{E}_x\left[\exp\left(\int^{T}_{0} \log g(s, \hat{X}_s) ds\right) g_T(\hat{X}_T) \right]\\
    &+\mathbb{E}_x\left[\exp\left(\int^{T}_{0} \log g(s, \hat{X}_s) ds\right) \left( g_T(\hat{X}_T) -\hat{g}_n(\hat{X}_T) \right) \right]\\
    &+\mathbb{E}_x\left[\hat{g}_n(\hat{X}_{T}) \Big(\prod_{i=0}^{n-1}\exp\big(\int_{i\eta}^{(i+1)\eta}\log g(s,\hat{X}_s)ds\big) - \prod_{i=0}^{n-1}\exp\big(\int_{i\eta}^{(i+1)\eta}\eta^{-1}\log \hat{g}_{i\eta}(\hat{X}_{i\eta})ds\big) \Big) \right],
\end{aligned}
\end{equation*}

For the first term above, by condition (d) in Assumption \ref{ass1}, $\exp(\int^{T}_{k\eta} \log g(s, X_s) ds) g_T(X_{T})$ is a continuous, bounded functional of $X_{[k\eta,T]}$ with $X_{k\eta} = x$, denote by $F_k(X_{[k\eta,T]})$. Then, using the KL upper bound for path measures obtained in \eqref{eq:aftergirsanov}, we have
\begin{equation*}
\begin{aligned}
    &\quad\mathbb{E}_x\left[\exp\left(\int^{T}_{0} \log g(s,X_s) ds\right) g_T(X_T) \right] - \mathbb{E}_x\left[\exp\left(\int^{T}_{0} \log g(s, \hat{X}_s) ds\right) g_T(\hat{X}_T) \right]\\
    &=\mathbb{E}_x\left[F_0(X_{[0,T]}) \right] - \mathbb{E}_x\left[F_0(\hat{X}_{[0,T]}) \right]\\
    &=\int F_0(y) \left(P_{[0,T]}^x(y) - \hat{P}_{[0,T]}^x(y) \right)dy\\
    &\leq C \,\text{TV}(P_{[0,T]}^x , \hat{P}_{[0,T]}^x)
    \leq C \DKL(\hat{P}_{[0,T]}^x \| P_{[0,T]}^x)^{\frac{1}{2}}
    \leq C\sqrt{Td\eta} \rightarrow 0,
\end{aligned}
\end{equation*}
where we have used Pinsker's inequality in the last line above.

For the second term, by conditions (b), (d) in Assumption \ref{ass1}, $\exp(\int^{T}_{0} \log g(s, \hat{X}_s) ds)$ is bounded, and $\hat{g}_n(x) \rightarrow g_T(x)$ uniformly in $x$. Therefore, as $\eta \rightarrow 0$,
\begin{equation*}
    \mathbb{E}_x\left[\exp\left(\int^{T}_{0} \log g(s, \hat{X}_s) ds\right) \left( g_T(\hat{X}_T) -\hat{g}_n(\hat{X}_T) \right) \right] \rightarrow 0.
\end{equation*}

For the third term, for $s \in [i\eta,(i+1)\eta)$,
\begin{multline*}
    |\log g(s,\hat{X}_s) - \eta^{-1}\log\hat{g}_{i\eta}(\hat{X}_{i\eta})| \leq |\log g(s,\hat{X}_s) - \log g(i\eta,\hat{X}_{s})|\\ + |\log g(i\eta, \hat{X}_s) - \eta^{-1}\log\hat{g}_{i\eta}(\hat{X}_{s})| + |\eta^{-1}\log\hat{g}_{i\eta}(\hat{X}_s) - \eta^{-1}\log\hat{g}_{i\eta}(\hat{X}_{i\eta})|.
\end{multline*}
By conditions (b), (c) in Assumption \ref{ass1}, as $\eta \rightarrow 0$, $|\log g(s,x) - \log g(i\eta,x)|$, $|\log g(i\eta,x) - \eta^{-1}\log\hat{g}_{i\eta}(x)|$ tend to 0 uniformly in $s$, $i$, $x$. Moreover, by condition (c) in Assumption \ref{ass1}, $|\log\hat{g}_{i\eta}(\hat{X}_s) - \log\hat{g}_{i\eta}(\hat{X}_{i\eta})| \lesssim |\hat{X}_{s} - \hat{X}_{i\eta} | \lesssim \eta |b(\hat{X}_{i\eta})| + |W_s - W_{i\eta}| $.
Hence, as $\eta \rightarrow 0$,
\begin{equation*}
    \mathbb{E}_x\left[\hat{g}_n(\hat{X}_{T}) \Big(\prod_{i=0}^{n-1}\exp\big(\int_{i\eta}^{(i+1)\eta}\log g(s,\hat{X}_s)ds\big) - \prod_{i=0}^{n-1}\exp\big(\int_{i\eta}^{(i+1)\eta}\eta^{-1}\log \hat{g}_{i\eta}(\hat{X}_{i\eta})ds\big) \Big) \right] \rightarrow 0.
\end{equation*}
Combining all the above, we conclude that for any fixed $x \in \mathbb{R}^d$,
\begin{equation*}
    | Z_{dis}(x) - Z_{con}(x) |  \rightarrow 0 \quad \text{as} \quad \eta \rightarrow 0.
\end{equation*}
\end{proof}

\bibliographystyle{plain}
\bibliography{main}

\end{document}